\documentclass[11pt]{article}

\usepackage{macros_arxiv}

\usepackage[protrusion=true,expansion=true]{microtype}

\title{\huge Finding Regularized Competitive Equilibria of Heterogeneous Agent Macroeconomic Models with Reinforcement Learning}
\author{Ruitu Xu\thanks{Yale University. E-mail: {\tt ruitu.xu@yale.edu}} \qquad\qquad ~
Yifei Min\thanks{Yale University. E-mail: {\tt yifei.min@yale.edu}} \qquad\qquad ~
Tianhao Wang\thanks{Yale University. E-mail: {\tt tianhao.wang@yale.edu}}\\
Zhaoran Wang\thanks{Northwestern University. E-mail: {\tt zhaoranwang@gmail.com}}\qquad\qquad
Michael I. Jordan\thanks{University of California, Berkeley. E-mail: {\tt jordan@cs.berkeley.edu }}\qquad\qquad
Zhuoran Yang\thanks{Yale University. E-mail: {\tt zhuoran.yang@yale.edu}}
}
\date{}

\begin{document}

\maketitle

\begin{abstract}
    We study a heterogeneous agent macroeconomic model with an infinite number of households and firms competing in a labor market.
    Each household earns income and engages in consumption at each time step while aiming to maximize a concave utility subject to the underlying market conditions.
    The households aim to find the optimal saving strategy that maximizes their discounted cumulative utility given the market condition, while the firms determine the market conditions through maximizing corporate profit based on the household population behavior. 
    The model captures a wide range of applications in macroeconomic studies, and we propose a data-driven reinforcement learning framework that finds the regularized competitive equilibrium of the model. 
    The proposed algorithm enjoys theoretical guarantees in converging to the equilibrium of the market at a sub-linear rate.
\end{abstract}

\section{Introduction}\label{sec:intro}

The behavior of labor markets has always been one of the key subjects of study in macroeconomics, and it is crucial to understand the underlying mechanisms that give rise to the aggregate macroeconomic indicators, such as price level and unemployment rate. 
In addition, detailed economic statistics such as wealth distribution have also become critical factors in analyzing the societal impact of economic policies.
There is therefore a major need for a deeper understanding of the interaction between economic variables in large and complex markets, a challenging endeavor for economists that machine learning and ubiquitous data collection promise to catalyze.

A variety of macroeconomic models have been proposed to characterize high-dimensional economic dependencies~\citep{aiyagari1994uninsured,bewley1986stationary,huggett1993risk}. 
Under these models, the competition in the market among different types of market participants, such as households and firms, is often regarded as a sequential game, where the participants interact and eventually arrive at certain equilibrium strategies \citep{kuhn2013recursive}. 
Being able to find these equilibrium outcomes is critical in unraveling many important aspects of economic growth and providing insightful guidance to policymakers and corporations.

In this paper, we focus on a labor market model that is composed of two groups of participants: households and firms, each with an infinite number of heterogeneous agents. 
One of the classic examples of such heterogeneous agent model is Aiyagari model \citep{aiyagari1994uninsured}.
At every time step, each household retains certain capital holdings, \eg, cash and investments, and earns a certain amount of income, \eg, salary and dividends. 
The income of each household is subject to independent and heterogeneous exogenous shocks over time, which reflects the impact of market conditions on individual household employment. 
Facing the capital holdings and income at each step, the household saves part of its total wealth for the next step and spends the rest to enjoy some utility.

The competition between the households and firms lies at the essence of 
the model:
The goal of the household is to come up with a saving strategy that maximizes its discounted cumulative utility. 
The utility function is concave with respect to savings, a critical assumption that reflects the economic principle of diminishing returns \citep{shephard1974law}. 
The firms make decisions that maximize corporate profit given the population behavior of the households subject to the saving strategy, and these decisions then determine the market conditions. 
Such competition between the two parties is characteristic of modern labor markets, which eventually arrive at an equilibrium determined by the inherent properties of the markets.
Note that the model comprises infinitely many households interacting through a population variable, and it can be cast into a \MFG \citep{lasry2007mean}, such that the representative agent of the \MFG stands for a typical household and the mean-field term represents the market condition that incorporates both the population behavior of the households and the optimal decisions of a typical firm~\citep{light2022mean}.

Existing numerical methods in economic literature suffer from a number of drawbacks in solving for such equilibria, and these drawbacks impose significant limits on the applications of the model to real markets \citep{kuhn2013recursive,achdou2014heterogeneous,achdou2022income,achdou2014partial}.
Specifically:
1) The underlying stochastic model for income shocks is unknown in general,
while existing methods assume prior knowledge of income shock transitions, which 
is unrealistic in practice.
2) These methods also require a discretization of the state-action space in 
computing the utilities despite the model featuring a continuous variable space. 
It results in massive computational inefficiency due to the curse of dimensionality~\citep{bellman2015applied}.
3) There are few theoretical results that guarantee the convergence of these 
existing methods to the desired equilibrium.

In this work we present a \RL based alternative to the existing methods.
Machine learning techniques
have been adopted for solving economic models recently, thanks to the emergence of massive computational power and the increasingly available microeconomic data \citep{achdou2022income,curry2022finding,min2022learn}. 
Despite some attempts to apply data-driven \RL methods for macroeconomic models
\citep{curry2022finding}, theoretical understanding of such algorithms is still limited. 
In particular, the analysis of \RL algorithms for heterogeneous agent models face several unique challenges:
1) The concavity assumption on the utility function of the households features the economic nature of the problem, and it induces a concave shape constraint on the value functions of the agents for any feasible policy.
Such shape constraint is absent in the existing literature on MFG, and it requires special treatments to achieve better sample efficiency. 
2) The model induces a continuous domain of the value functions, whereas 
existing literature on MFG only considers a discrete state-action space. 
Under such a continuous setup of state-action space, the optimal policy of an 
agent may not be unique, and this may compromise the uniqueness of the 
equilibrium as well as the stability of the learning process.

\textbf{Our contributions.}
We summarize the main contributions of this paper as follows:

\begin{itemize}
    \item 
    We propose a \RL framework for a class of heterogeneous agent models 
    in macroeconomics, by formulating it as an \MFG with a shape constraint on 
    the value function.
    Our formulation generalizes the well-known Aiyagari model, extending the model into high-dimensional state-action space.
    
    \item To guarantee a provable convergence to an equilibrium, we propose a regularized policy iteration algorithm that we refer to 
    as \texttt{ConcaveHAM}, facilitated by the combination of fitted $Q$-iteration
    \citep{riedmiller2005neural} and convex regression.
    Our algorithm provides a data-driven approach that estimates the income shock transitions, and with the introduction of convex regression, it avoids the curse of dimensionality in discretization and provides a continuous solution on value functions that better characterize the problem.
    It is noteworthy that our algorithm can readily incorporate other forms of 
    shape constraints under large state spaces. 
    
    \item We prove that \texttt{ConcaveHAM} converges to the desired equilibrium at a sublinear rate under standard assumptions. 
    Especially, it learns a \QRMFE of the \MFG, an equivalence of the regularized recursive competitive equilibrium of the macroeconomic model \citep{prescott2005recursive}, which incorporates the concept of bounded rationality \citep{selten1990bounded} and enforces the uniqueness of optimal policy through entropy regularization.
    To the best of our knowledge, this is the first data-driven framework for heterogeneous agent models that captures economic intuitions and guarantees theoretical tractability simultaneously.
    
    \item In the MFG that we consider, the feasible action set is not independent of its current state, which is a natural consequence of the household budget constraint. This is also a unique characteristic that has not been considered in the existing literature to the best of our knowledge and might be of independent interest.
\end{itemize}

\noindent\textbf{Notation.}
We denote the set of all positive integers by $\ZZ_+$ and the set of all non-negative real numbers by $\RR_+$.
For functions $f$ and $g$, we denote $f(n)\lesssim g(n)$ if $f(n)\leq Cg(n)$ for every $n\in\ZZ_+$ with some universal constant $C>0$ (similarly, $f(n)\gtrsim g(n)$); further, we write $f(n)\asymp g(n)$ if $f(n)\lesssim g(n)$ and $f(n)\gtrsim g(n)$. 
We write $\diam(\calX) \defeq \max_{x,y\in\calX}\|x-y\|_\infty$ for any set $\calX$.
For any measurable set $A$, we define $\Delta_A$ as the set of all density functions supported on $A$.
We denote the set of policies $\pi$ with $\Pi$.
For any measurable functions $f$ and $g$ supported on domain $A$, we denote $\inner{f}{g}_A = \int_A f(x)g(x)\diff x$.

\section{Preliminary}\label{sec:preliminaries}

In this paper, we focus on a heterogeneous agent model as a generalization of the well-known Aiyagari model of labor markets with idiosyncratic income shocks \citep{aiyagari1994uninsured}.
Notably, we formulate the model as an \MFG, \ie, an asymptotic approximation of a multi-agent Markov game with a large number of agents. 
In the rest of this section, we elaborate on how to cast a generalized Aiyagari model in macroeconomics into the form of an \MFG. \cref{tab:comparison} shows a detailed correspondence between the two interpretations.

\subsection{Heterogeneous agent model}
\label{sec:equivalence}

\begin{table*}[t]
\caption{A comparison of notation between the generalized Aiyagari model and its \MFG formulation.}
\label{tab:comparison}
\centering
\begin{tabular}{llll}
    \toprule
    \multicolumn{2}{c}{Mean-field game} & \multicolumn{2}{c}{Heterogeneous agent model} \\
    \multicolumn{2}{c}{(Representative agent)} & \multicolumn{2}{c}{(Representative household)} \\
    \cmidrule(r{4pt}){1-2} \cmidrule(l){3-4}
    Notation & Interpretation & Notation & Interpretation \\
    \midrule
    $s\in\calS\subseteq\RR^{d_s}$ & State & $(b,w)\in\calB\times\calW \subseteq \RR^{d_s}$ & Capital holdings and income \\
    $a\in\calA\subseteq\RR^{d_w}$ & Action & $a\in\calA\subseteq\RR^{d_w}$ & Capital savings \\
    $\pi:\calS\to\Delta_\calA$ & Policy & $\pi:\calS\to\Delta_\calA$ & Saving strategy \\
    $z\in\calZ\subseteq\RR^{d_z}$ & Mean-field term & $z\in\calZ\subseteq\RR^{d_z}$ & Market condition (indicators) \\
    $r_z$ & Reward function & $r_z$ & Utility function \\
    $V_z^\pi$ & Value function & $V_z^\pi$ & Cumulative utility \\
    $\Xi_z^\pi$ & Aggregate indicator & $(b_z^\pi,w_z^\pi)$ & Population capital and labor \\
    $(z^\star,\pi^\star)$ & \QRMFE & $(z^\star,\pi^\star)$ & Market equilibrium \\
    \bottomrule
\end{tabular}
\end{table*}

Consider $d_b,d_w,d_z\in\ZZ_+$.
Let $\calB\subseteq [0,1]^{d_b}$ be the set of all capital holdings, 
$\calW\subseteq [0,1]^{d_w}$ be the set of all possible incomes,
and $\calA \subseteq [0,1]^{d_w}$ be the set of all feasible capital savings.
Further let $\calZ\subseteq\RR^{d_z}$ be the set of all aggregate indicators 
that represent possible market conditions.

\paragraph{Households.}
Within a heterogeneous agent model, each household on the market is characterized by its total assets $b\in\calB$ and current income $w\in\calW$ at each time step.
Under any given market condition represented by some aggregate indicators $z\in\calZ$, the income $w$ of each household is subject to independent heterogeneous exogenous shocks, \ie, the income of each household forms a Markov chain as time evolves. 
The household adopts a saving strategy $\pi:\calB\times\calW\to\calA$, and at each time step, after making saving $a=\pi(b,w)$, the household gains utility $r_z(b,w,a)$, where $r_z:\calB\times\calW\times\calA\to\RR$ is the utility function under market condition $z$.
The savings $a$ made at the current time step then become the capital holdings $b$ at the next time step.

The goal of each household is to come up with a saving strategy $\pi_z^\star$ 
that maximizes its expected cumulative utility $V_z^{\pi_z^\star}$ (as a $\gamma$-discounted sum of all future rewards) subject to the underlying market condition $z$. 
The optimal saving strategy of the households, together with the aggregate indicator, gives rise to a population average of household capital retention $b_z^{\pi^\star}$ and income level $w_z^{\pi^\star}$.

\paragraph{Firms.}
Based on $b_z^{\pi^\star}$ and $w_z^{\pi^\star}$, the firms desire to make decisions that maximize their profits, which in turn gives a new market condition indicator $z'$ following some aggregate function $\Psi$ and production mapping $\Phi$. 

More specifically, for $d_\Xi\in\ZZ_+$ and the set $\Pi$ of all feasible policies, we assume the representative firm has access to an aggregate function $\Psi:\calZ\times\Pi\to\RR^{d_\Xi}$ that maps any mean-field term $z$ and policy $\pi$ to a set of aggregate indicators $\Xi_z^\pi = (b_z^\pi,w_z^\pi)\in\RR^{d_\Xi}$ (which represents the population behaviors of all households in the market). 
Here $b_z^\pi$ denotes the aggregated indicators on capital retention and $w_z^\pi$ denotes those on labor supply.
The representative firm then makes corporate decisions based on the indicators $\Xi_z^\pi$.
In particular, the firm takes $\Xi_z^\pi$ as given and picks the corporate decisions that maximize the production function of the firm, which give rise to a new market condition represented by a new mean-field term $z' = \Phi(\Xi_z^\pi)\in\calZ$.

\paragraph{Competition between households and firms.}
Under the new market condition $z'$, the households then need to update their saving policy, and the competition between the households and the firms continues iteratively.

Notice that the households and firms interact only through the aggregate terms $b_z^\pi$, $w_z^\pi$, and $z$; therefore, the competition can be regarded as between one representative household and one representative firm.
Such repetitive competitions reach a competitive equilibrium $(z^\star,\pi_{z^\star}^\star)$ under mild conditions. Our goal is to learn an approximate equilibrium of $(z^\star,\pi_{z^\star}^\star)$ from observational data, where the behavior policy may not be optimal with respect to the underlying market condition $z$.

\subsection{Mean-field game on households}

We cast the above heterogeneous agent model into a framework of \MFG, where infinitely many identical agents interact through a mean-field term \citep{lasry2007mean}. 
In particular, we focus on a representative agent that stands for a representative household and consider its interaction with a mean-field term $z$, which forms a counterpart of the aggregate indicators $z$ of the market.
From now on, we describe the components of the MFG in the language of RL, as listed in \cref{tab:comparison}.

\paragraph{Shape-constrained MDP for representative agent.}
Fix any mean-field term $z$.
The interaction between the agent and economic
environment forms a discounted infinite-horizon MDP with a shape constraint, 
denoted by a tuple $(\calS, \calA, \calG, \Gamma, \PP_z, r_z, \gamma)$, 
where $\gamma\in(0,1)$ is the discount factor.
Here $\calS \defeq \calB\times\calW$ is the state space of total assets and income, 
and $\calA$ is the action space of savings.
In contrast to standard definitions of MDPs, here at each state $s\in\calS$, 
the agent can only take action in a feasible actions set 
$\Gamma(s)\subseteq\calA$ that reflects the household budget constraint.
We define the feasible state-action set as $\calG\defeq\{(s,a)\in\calS\times\calA 
\mid a \in \Gamma(s)\}$. 
In particular, we assume that $\calG\subseteq\RR^{d_s+d_w}$ is a convex set,
corresponding to a convex household budget constraint, and that it is 
independent of the choice of mean-field term $z$.
We also assume the reward function $r_z:\calG\to[0,1]$ is a concave function 
that depends on the mean-field term $z$.

At each time step $h$ given state $s_h=(b_h,w_h)$, the agent takes some action 
$a_h \in \Gamma(s_h)$ and receives a reward $r_z(s_h,a_h)$. 
Then the agent transitions to the next state $s_{h+1}=(b_{h+1},w_{h+1}) 
\sim \PP_z(\cdot\mid s_h,a_h)$, 
where the transition probability $\PP_z$ represents the idiosyncratic income shock that implicitly depends on the mean-field term $z$.
More specifically, for any $a_h\in\Gamma(b_h,w_h)$.
\begin{align}\label{eq:transition}
    \PP_z(b_{h+1},w_{h+1}|b_h,w_h,a_h) = \begin{cases}
    \prob_z[w_{h+1}|w_h] & b_{h+1}=a_h\\
    0 & b_{h+1}\neq a_h
    \end{cases}
\end{align}
where $\PP_z[\cdot\given\cdot]$ denotes the income shock.

\paragraph{Learning goal.}
The representative agent aims to maximize the $\gamma$-discounted cumulative 
reward over an infinite time horizon.
We encode the agent's strategy in a map $\pi:\calS\to\Delta_\calA$, which is 
called a \emph{policy}.
In particular, $\pi(s)$ is only supported on $\Gamma(s)$ for any $s\in\calS$.
Then given any mean-field term $z$ and policy $\pi$, we define the Q-function (\ie, action-value function) $Q_z^\pi:\calG\to \RR$ as the expected discounted cumulative reward 
under the policy $\pi$, \ie, 
\begin{align*}
    Q_z^\pi(s,a) \defeq \expect\bigg[\sum_{h=1}^\infty \gamma^h r_z(s_h,a_h)\ \bigg|\ s_0=s,a_0=a\bigg],
\end{align*}
where the expectation $\expect$ is taken with respect to both $\PP_z$ and $\pi$
along the trajectory.
Similarly, we define the value function $V_z^\pi(s):\calS\to\RR$ of $\pi$ as 
\begin{align*}
    V_z^\pi(s) \defeq \expect\bigg[\sum_{h=0}^\infty \gamma^h r_z(s_h,a_h)\ \bigg|\ s_0=s\bigg],  
\end{align*}
which we also write as $V_z^\pi(s) = \inner{Q_z^\pi(s,\cdot)}{\pi(\cdot \given s)}_{\Gamma(s)}$. 
The goal of the representative agent is then to 
find an optimal policy $\pi_z^\star$ that maximizes its value functions 
$Q_z^\star(s,a) \defeq \max_\pi Q_z^\pi(s,a)$ and 
$V_z^\star(s) \defeq \max_\pi V_z^\pi(s)$.
Note that such optimal policy always exists for infinite horizon \MDPs \citep{puterman2014markov}, and it holds that $V_z^*(s) = \inner{Q_z^*(s,\cdot)}{\pi_z^*(\cdot \given s)}_{\Gamma(s)}$.

Note that the convexity of the reward functions further induces the 
convexity of the value functions for any policy, and thus the policies are 
economically meaningful for the heterogeneous agent model.

\paragraph{Mean-field term: the representative firm.}
In our formulation of the \MFG, we treat the representative firm as a function 
that outputs the mean-field term given any policy $\pi$ from the representative agent.

Formally, we assume a mapping $\Phi\circ\Psi:\calZ\times\Pi\to\calZ$ that sends any policy $\pi$ under a current mean-field term $z$ to its corresponding updated mean-field term $z'$.
We assume full knowledge of both functions, as the aggregation function $\Psi$ can be estimated with a simulator and the production mapping $\Phi$ is generally a deterministic function given the production function.

\subsection{Quantal response mean-field equilibrium}

The \MFG described above admits at least one equilibrium under mild conditions, and any equilibrium $(z^\star,\pi^\star)$ of the \MFG then corresponds to an equilibrium of the heterogeneous agent model, where the equilibrium reveals itself as the pair of market condition $z^\star$ and the optimal household strategy $\pi_{z^\star}^\star$ under such market condition.

\textbf{Issues with RCE.}
One commonly considered equilibrium for \MFG is \RCE, which is represented as a 
tuple $(z^\star,Q^\star,\pi^\star,\Xi^\star)$ where $\pi^\star$ is the optimal 
policy with respect to the MDP under transition kernel $\PP_{z^\star}$, the 
optimal Q-function $Q^\star = Q_{z^\star}^\star$, and it holds that
$\Phi(\Xi^\star) = z^\star$ for the aggregate indicators 
$\Xi^\star = \Psi(z^\star,\pi^\star)$ \citep{light2022mean}. 
However, such equilibrium can be non-unique and unstable with respect to $Q^\star$ and its estimation. 
Furthermore, policy estimation is not robust under the definition of \RCE.
The lack of robustness limits the capability of value-based \RL algorithms from converging to the equilibria.
Further discussions on the lack of uniqueness and robustness of \RCE is delayed to \cref{sec:issues-rce}.

To overcome this challenge, we propose a regularized competitive equilibrium named Quantal Response Mean-Field Equilibrium (QR-MFE).

\begin{definition}[Quantal Response Mean-Field Equilibrium]\label{def:ne}
    A representative agent policy and a mean-field term $(\pi^\star,z^\star)$ reach quantal response mean-field equilibrium with respect to a strongly convex entropy function $\calH$ if $z^\star = \Phi(\Psi(z^\star,\pi^\star))$ and for any policy $\pi$, state $s\in\calS$, and the optimal Q-function $Q_{z^\star}^\star$ under mean-field term $z^\star$
        \begin{align*}
            & \int_{\Gamma(s)} \pi^\star(a\given s)Q_{z^\star}^\star(s,a) \diff a - \calH(\pi^\star) \geq \int_{\Gamma(s)} \pi(a\given s)Q_{z^\star}^\star(s,a) \diff a - \calH(\pi).
        \end{align*}
\end{definition}

The quantal response equilibrium introduces the notion of bounded rationality that characterizes human decision-making, where the agent makes stochastic ``quantal'' decisions according to a smooth probability distribution around the best response \citep{goeree2020stochastic}.

Under this new notion of equilibrium, the equilibrium policy $\pi^\star$ is unique with respect to $Q_{z^\star}^\star$ thanks to the strongly convex entropy regularization $\calH$. 
It is noteworthy that the pair $(\pi^\star,z^\star)$ of policy and mean-field term alone is enough to characterize the equilibrium; the optimal Q-function $Q_{z^\star}^\star$ is deterministic given $z^\star$ and the aggregation indicators $\Xi^\star = \Psi(z^\star,\pi^\star)$ are also fixed given $\pi^\star$ and $z^\star$. Especially, the optimal policy can be estimated from data through \RL algorithms when the mean-field term is fixed, and therefore we only consider the approximation error to the equilibrium mean-field term $z^\star$ in the rest of this paper.

\section{A Reinforcement Learning Algorithm}

In this section, we propose a RL-based algorithm called \texttt{ConcaveHAM} for 
learning the \QRMFE of the heterogeneous agent model from offline data.
The main algorithm is displayed in \cref{algo:iteration}, which, at a high level, 
consists of two major procedures: 1) a mean-field term generation from regularized 
optimal policy, and 2) an optimal value function estimation under concave shape constraint.

\cref{algo:iteration} executes as follows:
We initialize \cref{algo:iteration} with an arbitrary mean-field term $z^0\in\calZ$, zero Q-function $Q^0$, and uniform policy $\widehat\pi^0$, which is indeed the regularized optimal policy of constant Q-function $Q^0$.
In each iteration $t\in[T]$, given the previous mean-field term $z^{t-1}$ and the representative agent's policy $\widehat\pi^{t-1}$, the representative firm determines a new mean-field term $z^t$ (\cref{algline:z}).
Then under $z^t$, we estimate the optimal Q-function $\widehat Q^t$ for the representative agent with the concave \FQI~(\cref{algo:q-iteration}), using the offline dataset~(\cref{algline:Q}).
The regularized optimal policy is further solved according to \eqref{eqn:regularized-optimal} with respect to $\widehat Q^t$ (\cref{algline:pi}).
Finally, \cref{algo:iteration} returns an approximation $z^T$ of the equilibrium mean-field term $z^\star$ after $T$ iterations.

As we will show in Theorem~\ref{thm:contraction}, the mean-field term estimation of \cref{algo:iteration} converges to the equilibrium quantity $z^\star$ under proper conditions.
In the rest of this section, we explain the details for each component of \cref{algo:iteration}.

\begin{algorithm}[tb]
\centering
\caption{Concave Heterogeneous Agent Model (\texttt{ConcaveHAM})}\label{algo:iteration}
    \begin{algorithmic}[1]
        \Require Mean-field term $z^0\in\calZ$, number of iterations $T$, number of iterations $\tau$ within each \texttt{CFQI}
        \State {\bf initialize} $Q^0(s,a)\gets 0$ for all $(s,a)\in\calG$, $\widehat\pi^0\gets \mathsf{Unif}(\Gamma(s))$ for all $s\in\calS$
        \For{$t = 1,\ldots,T$}
            \State $z^t \gets \Phi(\Psi(z^{t-1},\widehat\pi^{t-1}))$ \label{algline:z}
            \State $\widehat Q^t \gets \texttt{CFQI}(\calD_t;\tau)$ where $\calD_t$ is the offline dataset under $z^t$ \label{algline:Q}
            \State $\widehat\pi^t\gets \pi_{\widehat Q^t}$ following \cref{eqn:regularized-optimal} \label{algline:pi} 
        \EndFor
        \State Return $z^T$ \label{algline:z}
    \end{algorithmic}    
\end{algorithm}

\paragraph{Regularized optimal policy.}
For any Q-function $Q$ and any strongly convex regularizer $\calH$, we define the regularized optimal policy $\pi_Q$ with respect to $Q$ as
\begin{align}\label{eqn:regularized-optimal}
    \pi_Q(\cdot \given s) \defeq \argmax_{u \in \Delta_{\Gamma(s)}} \bigg\{\int_{\Gamma(s)} u(a)Q(s,a) \diff a - \calH(u)\bigg\}.
\end{align} 
In particular, we assume $\calH$ is $\zeta$-strongly convex in $\|\cdot\|_1$. 
A classic example is the negative entropy, i.e., $-\calH(u) = \zeta \int u(a)\log u(a) \diff a$. 
We further assume $\pi_Q$ is known for any Q-function $Q$ and regularizer $\calH$ for simplicity, and the policy optimization with entropy regularization can be found in \citet{mei2020global,chen2021primal}.

\textbf{Q-function estimation.}
For any fixed mean-field term $z$, we propose a Concave \FQI algorithm in \cref{algo:q-iteration}, abbreviated as \texttt{CFQI}, for estimating the optimal Q-function $Q_z^\star$ of the underlying \MDP from data. 

Under each mean-field term $z$, we define the Bellman optimality operator $\calT_z$ as follows:
\begin{align*}
    (\calT_z Q)(s,a) \defeq r_z(s,a) + \gamma \expect_{s'\sim\PP_z(\cdot\mid s,a)}[(\frakJ Q)(s')],
\end{align*}
where $(\frakJ Q)(s) \defeq \max_{a\in \Gamma(s)} Q(s,a)$.
For each mean-field term $z\in\calZ$, we assume access to an offline dataset $\calD_z \defeq \{(s_{z,m},a_{z,m},r_{z,m},s_{z,m}')\}_{m=1}^M$ that contains $M$ \iid samples, where $(s_{z,m},a_{z,m})\sim\widetilde\mu\in\Delta_\calG$, the reward $r_{z,m} = r_z(s_{z,m},a_{z,m})$, and the next state $s_{z,m}'\sim \PP_z(\cdot\given s_{z,m},a_{z,m})$.
Then given the dataset $\calD_z$,
\FQI \citep{antos2007fitted} applies an approximate value iteration on the Q-functions, by iteratively computing an estimator $\widetilde Q^\ell$ of $\calT_z \widetilde Q^{\ell-1}$ for $\ell=1,\ldots,\tau$.

In particular, \texttt{CFQI} enforces a concave shape constraint on the estimated Q-function through a convex regression solver termed \LSE.
Given $\widetilde Q^{\ell-1}$ from the previous iteration and a set of concave functions denoted by $\scrF$, LSE finds $\hat f\in\cal F$ that minimizes the empirical risk given by the Bellman error between $\hat f$ and $\widetilde Q^{\ell-1}$ evaluated on the dataset $\calD_z$.
Specifically, for any $Q$ and $\widehat f\in\scrF$, we say $\widehat f$ is an $\epsilon$-approximate \LSE of $\calT_z Q$ on dataset $\calD_z$ if
\begin{align*}
    & \sum_{(s,a,r,s')\in\calD_z} (\widehat f(s,a) - r - \gamma (\frakJ Q)(s'))^2 - \min_{f\in\scrF} \sum_{(s,a,r,s')\in\calD_z} (f(s,a) - r - \gamma (\frakJ Q)(s'))^2 \leq \epsilon.
\end{align*}
We denote the 
regression solver used to acquire such $\epsilon$-approximate \LSE as $\texttt{LSE}(\scrF,\calD_z,\calT_z Q;\epsilon)$.
Then in each iteration $\ell$, we solve for $\widetilde Q^\ell = \texttt{LSE}(\scrF,\calD_z,\calT_z\widetilde Q^{\ell-1};\epsilon)$, and \cref{algo:q-iteration} terminates after $\tau$ iterations and returns an approximate optimal Q-function $\widetilde Q^\tau$.

\begin{algorithm}[tb]
\centering
\caption{Concave \FQI (\texttt{CFQI})}\label{algo:q-iteration}
    \begin{algorithmic}[1]
        \Require Data $\calD_z = \{s_{z,m},a_{z,m},r_{z,m},s'_{z,m}\}_{m=1}^M$, number of iterations $\tau$
        \State {\bf initialize} $\widetilde Q^0(s,a) \gets 0$ for all $(s,a)\in\calG$
        \For{$\ell = 1,\ldots,\tau$}
            \State $\widetilde Q^\ell \gets \texttt{LSE}(\scrF,\calD_z,\calT_z \widetilde Q^{\ell-1};\epsilon)$ 
        \EndFor
        \State Return $\widetilde Q^\tau$
    \end{algorithmic}    
\end{algorithm}

In general LSE is hard to obtain without structures on $\scrF$, but fortunately here we have the convex shape constraint.
Recall that the reward function $r_z$ is concave and takes values in $[0,1]$, and it follows that 
$Q_z^\pi$ and $V_z^\pi$ are in $[0,1/(1-\gamma)]$.
Let us denote $B \defeq 1/(1-\gamma)$ and $L \defeq L_r/(1-\gamma)$, then the Q-function $Q_z^\pi$ is $B$-bounded and $L$-Lipschitz for any $z\in\calZ$ and policy $\pi$.

Hence, as \cref{lem:concavity-general} shows, the estimation target $\calT_z Q$ is a concave function for any concave $Q$, and this allows effective estimation of $\calT_z Q$ with a simple concave function set. 
We also remark that if the model has finite number of discrete income levels for variable $w$, the regression target $\calT_z Q(\cdot,w,\cdot)$ is concave for any concave $Q(\cdot,w,\cdot)$ and all $w\in\calW$ without the requirement of stochastic concavity, \cf \cref{lem:concavity}.

\subsection{Convex regression}
We now introduce a few regression solvers available to \cref{algo:q-iteration}.
Without loss of generality, 
for any regression on a concave $Q$ bounded from above by $B$, we consider an equivalent regression problem on the non-negative convex function $-Q + B$. 
With a slight abuse of notation, we write $Q$ to denote this transformed convex function from now on to simplify discussions.

\paragraph{Bounded Lipschitz convex functions.}
Let $\scrC_{B,L}$ denote the set of all non-negative $B$-bounded and $L$-Lipschitz convex functions on $\calG$, \ie,
\begin{align*}
    & \scrC_{B,L} \defeq\{f: \calG \to [0,B] \mid f \text { is convex}; \partial f(x)\neq\emptyset, \forall x\in\calG; \|\partial f(x)\| \leq L, \forall x\in\calG\}, 
\end{align*}
where $\partial f(x)$ denotes the subdifferential of $f$ at $x \in \calG$ and $\|\partial f(x)\|$ denotes $\sup_{g\in \partial f(x)}\|g\|_\infty$.
Recall that $\calT_z Q\in\scrC_{B,L}$ for any $Q\in\scrC_{B,L}$ and mean-field term $z$ following \cref{lem:concavity-general}.
A convex regression solver for $\calT_z Q$ is written as $\texttt{LSE}(\scrC_{B,L},\calD_z,\calT_z Q;\epsilon)$. 

\paragraph{Lipschitz max-affine functions.}

Further, we let $\scrA_{B,L}^{K,+}$ denote the set of all non-negative $L$-Lipschitz $K$-max-affine functions, \ie,
\begin{align*}
    & \scrA_{B,L}^{K,+} \defeq \{h:\calG\to[0,B] \given h(x) = \max_{k\in[K]} \alpha_k\transpose x + c_k, \|\alpha_k\|_\infty \leq L, \forall k\in[K]\}.
\end{align*}
Note that max-affine functions are convex by definition, and therefore any estimator in $\scrA_{B,L}^{K,+}$ also lies in $\scrC_{B,L}$.
We write a $K$-max-affine regression solver for $\calT_z Q$ as $\texttt{LSE}(\scrA_{B,L}^{K,+},\calD_z,\calT_z Q;\epsilon)$.
It is worth noting that an approximate \LSE to the convex regression on both $\scrC_{B,L}$ and $\scrA_{B,L}^K$ can be obtained by solving the corresponding convex optimization problems up to arbitrary accuracy \citep{mazumder2019computational,balazs2015near}, and we therefore assume the approximate \LSEs returned by the solvers admit at most $\epsilon$ excess empirical risk compared to the \LSE.

\paragraph{Input convex neural networks.}
Beyond the standard convex regression algorithms, the shape-constrained estimation of $\calT_z Q$ can be achieved by minimizing the empirical risk over a set of \ICNN, where the network output is always convex with respect to the network input when the parameter space of the network is constrained \citep{amos2017input}. 

More specifically, for any $K$-layer \ICNN $f_K(x):\Rd\to\RR$, each of its \ith layer $y_i$ can be expressed as
\begin{align*}
    y_{i+1} = \sigma(W_i^{(y)} y_i + W_i^{(x)} x - \beta_i), \quad i = 0,1,\ldots,K-1,
\end{align*}
where $y_i,\beta_i\in\RR^{\ell_i}$, $W_i^{(y)}\in\RR^{\ell_{i+1}\times\ell_i}$, $W_i^{(x)}\in\RR^{\ell_{i+1}\times d}$, and $\sigma$ denotes ReLU function. 
Especially, $y_0 = x\in\Rd$ represents the network input, $y_K = f_K(x)\in\RR$ represents the network output, and $\ell_i$ denotes the number of neurons on the \ith layer (especially, $\ell_0 = d$). 
To guarantee the convexity of $f_K(x)$ with respect to $x$, we further restrict $W_i^{(y)}\in\RR_+^{\ell_{i+1}\times\ell_i}$, so that the composition of convex and convex non-decreasing function is also convex \citep{amos2017input}. 

We further show that under a convex parameter set of the network specified by $\calG$, we are able to construct a set $\scrN_{B,L}^K$ of \ICNNs such that it is equivalent to the set of all non-negative $B$-bounded and $L$-Lipschitz $K$-max-affine functions, \ie, $\scrA_{B,L}^{K,+}$. 
Therefore, minimizing the empirical risk on $\scrN_{B,L}^K$ also serves as an estimation of the convex function $\calT_z Q$, assuming that the solver trains to the global minimum.
This is summarized in the following lemma, and see \cref{sec:exact-representation-proof} for a proof.

\begin{lemma}\label{lem:exact-representation}
    There exists a set $\scrN_{B,L}^K$ of $K$-layer \ICNNs, such that for any $L$-Lipschitz function $f_K:\calX\to [0,B]$ defined on a convex domain $\calX\subseteq\Rd$ that is the maximum of $K$ affine functions, $\scrN_{B,L}^K$ is able to represent $f_K$ exactly under a convex constraint set specified by $W_i^{(x)} \geq 0$, $\|\sum_{j=i-1}^{K-1} w_j\|_\infty \leq L$, and $\max_{x\in\calX} \sum_{j=i-1}^{K-1} w_j\transpose x + \sum_{j=i-1}^{K-1} b_j \leq B$
    for all $i\in[K]$, where $w_i \defeq w_{i,1} - w_{i,2}$ and $W_i^{(x)} = [w_{i,1}\transpose,w_{i,2}\transpose]\transpose$.
    In addition, we set $W_0^{(y)} = 0$ and $W_i^{(y)} = 1$ fixed for all $i\in[K-1]$.
    Especially, $\scrN_{B,L}^K$ is equivalent to the set of all non-negative $K$-max-affine functions that are $B$-bounded and $L$-Lipschitz.
\end{lemma}

\section{Theoretical Results}

Next, we show the existence and uniqueness of the \QRMFE, and provide theoretical guarantees to the convergence of \cref{algo:iteration} to the unique equilibrium.
Before diving into the main results, we make a few moderate assumptions on the \MFG, under which our theoretical results hold.

We first define formally the shape constrain assumption on the feasible action sets and reward functions, as well as some technical conditions required for the following theorems.
It is noteworthy that we consider $b$ and $w$ as continuous variables.

\begin{assumption}\label{assump:convex}
The set of all feasible mean-field terms $\calZ\subseteq\RR^{d_z}$ has bounded $\ell_1$-radius, \ie, $\max_{z,z'\in\calZ}\|z-z'\|_1 \leq Z$ for some $Z > 0$.
The feasible set $\calG=\{(s,a)\in\RR^{d_s+d_w}\mid a\in\Gamma(s)\}$ is a closed convex set.
The reward function $r_z(s,a)$ is concave in $(s,a)$ and $L_r$-Lipschitz in $\ell_\infty$ norm with respect to $(s,a,z)$, i.e., $|r_z(s,a)-r_{z'}(s',a')| \leq L_r 
\cdot\max\{\|s-s'\|_\infty,\|a-a'\|_\infty,\|z-z'\|_\infty\}$.
\end{assumption}

\begin{assumption}\label{assump:stochastic_concave}
For all $z\in\calZ$, the transition kernel $\PP_z$ of the \MDP is stochastically concave \citep{smith2002structural}, \ie, $q(b,w,a) \defeq \expect_{w'}[Q(b,w',a)\given w]$ is concave on for any concave $Q$.    
\end{assumption}

We assume that the transition kernel \eqref{eq:transition} is Lipschitz with respect to both the state and mean-field term.
Recall that the transition is independent of the current capital holdings $b$ and savings $a$ (see Equation~\ref{eq:transition}).
Such assumption admits smooth transition between between different mean-field environments.
\begin{assumption}\label{assum:lipschitz-transition}
    The transition kernel $\PP$ is Lipschitz in the sense that for any mean-field term $z,z'\in\calZ$ and $s,s'\in\calS$, there exists $L_\PP > 0$ such that $\|\PP_z(\cdot \mid w)-\PP_{z'}(\cdot \mid w')\|_1 \leq L_\PP \cdot(\|s-s'\|+\|z-z'\|_1)$.
\end{assumption}

To guarantee a stable iteration for \cref{algo:iteration}, we also impose a Lipschitz condition on the aggregation function $\Psi$ and the production mapping $\Phi$. 
Here for any two policies $\pi$ and $\pi'$, $\|\pi-\pi'\|_{\infty,\overline\nu}$ is defined as $\int_\calS \|\pi(\cdot\given s) - \pi'(\cdot\given s)\|_\infty \diff\nu(s)$ for any $\nu\in\Delta_\calS$.
\begin{assumption}\label{assum:lipschitz-mappings}
    The aggregate function $\Psi$ is $L_{\Psi}$-Lipschitz for all policy $\pi$ and mean-field term $z\in\calZ$ and the production mapping $\Phi$ is $L_{\Phi}$-Lipschitz, \ie,
    \begin{align*}
        \|\Psi(z,\pi) - \Psi(z',\pi')\|_1 & \leq L_\Psi (\|z-z'\|_1 + \|\pi - \pi'\|_{\infty,\overline\nu}), \\
        \|\Phi(\Xi) - \Phi(\Xi')\|_1 & \leq L_\Phi\cdot \|\Xi-\Xi'\|_1,
    \end{align*}
    for some policy-induced distribution $\overline\nu$.
    Hence, $\Phi(\Psi(\cdot,\cdot))$ is $L_F$-Lipschitz, where $L_F \defeq L_{\Psi} L_{\Phi}$.
\end{assumption}

As mentioned above, we assume that for any mean-field term $z$, we have access to a logged dataset $\calD_z$ comprising of \iid samples from the trajectory of an exploratory behavior policy, where the behavior policy satisfies a coverage assumption.
These are standard for \FQI in the RL literature \citep{antos2007fitted}, and the state-action pairs needed for the training dataset can be sampled from the heterogeneous households in the market.
\begin{assumption}\label{assum:concentrability}
    For any $z\in\calZ$, the offline data $\calD_z = \{(s_{z,m},a_{z,m},r_{z,m},s'_{z,m})\}_{m=1}^M$ follows an exploratory distribution $\widetilde\mu$,
    where \iid samples $(s_{z,m},a_{z,m})\sim \widetilde\mu$. 
    In particular, $\widetilde\mu = \widetilde\nu\times\pi_b$ follows from a sample distribution $\widetilde\nu$ on $\calS$ and an exploratory behavior policy $\pi_b$, such that for any mean-field term $z$, non-stationary policy $\pi$\footnote{A non-stationary policy may change over time, \ie, there may not exists a policy $\pi:\calS\to\Delta_\calA$ such that $\pi_t = \pi$ for all $t$.}, 
    step $h$, and feasible state-action pair $(s,a)$, there exists a constant $D$ such that
    \begin{align*}
        \frac{\mu_{z,h}^\pi(s,a)}{\widetilde{\mu}(s,a)} \leq D,
    \end{align*}
    where $\mu_z^\pi$ denotes the distribution on step $h$ over $\calS \times \calA$ induced by transition kernel $\PP_z$ and policy $\pi$. 
    Note that $\widetilde{\nu}(s) = \int_{\Gamma(s)} \widetilde{\mu}(s,a) \diff a$, and we have $\int_{\Gamma(s)} \mu_{z,h}^\pi(s, a) \diff a/ \widetilde\nu(s) \leq D$ for any $s \in \calS$.
\end{assumption}

Under these assumptions, we show the existence and uniqueness of the \QRMFE for the \MFG as well as the convergence guarantees of \cref{algo:iteration} with a range of external solvers for convex regression.

To ease the notation, let us introduce some problem-dependent quantities: 
Let $B \defeq 1/(1-\gamma)$, $L \defeq L_r/(1-\gamma)$, $d \defeq d_s+d_w$, and
\begin{align*}
    \kappa \defeq \frac{J L_F D^2}{\zeta} + L_F. 
\end{align*}
We further write 
\begin{align*}
C_d = B^\frac{d+8}{d+2} (d+1)^\frac{2}{d+2} (dL)^\frac{d}{d+2}
\end{align*}
and $C_d' = (d+1)B^4 + (B+dL)dL$.

First, we prove that the \MFG has a unique \QRMFE.
The essential idea is that if $\Phi(\Psi(\pi_{Q_z^*},z))$ a contraction and $\calZ$ is complete, then the equilibrium is unique.
See proof in \cref{sec:contraction-proof}.

\begin{theorem}[Unique QR-MFE]\label{thm:contraction}
    Under Assumptions \ref{assump:convex}, \ref{assump:stochastic_concave}, \ref{assum:lipschitz-mappings} and \ref{assum:concentrability}, further assume that $\kappa<1$.
    Then, there exists a unique quantal response mean-field equilibrium to the \gls*{MFG} with regularized policies.
\end{theorem}

Further, we show that \cref{algo:iteration} converges to the \QRMFE with a range of external solvers that provides approximate \LSE on the set of convex functions, max-affine functions, and \ICNNs. In particular, as we will see in \cref{thm:main-convex,thm:main-affine-nn}, the \RL framework \texttt{ConcaveHAM} enjoys sub-linear convergence rate to the equilibrium mean-field term $z^\star$ for all three convex regression solvers if the estimation error $\epsilon = 0$.
Proofs for \cref{thm:main-convex,thm:main-affine-nn} can be found in \cref{sec:main-convex-proof,sec:main-affine-nn-proof}.

\begin{theorem}[Convergence for $\scrC_{B,L}$]\label{thm:main-convex}
    Under Assumptions \ref{assump:convex}, \ref{assump:stochastic_concave}, \ref{assum:lipschitz-transition}, \ref{assum:lipschitz-mappings}, and \ref{assum:concentrability}, further assume that $\kappa<1$.
    For $\delta\in(0,1)$, suppose that the number of iterations within each \texttt{CFQI} subroutine satisfies 
    \begin{align*}
        \tau \asymp \frac{\log\left( C_d^{-1} D^{-1/2} \frac{M^{2/(d+2)}}{\log M} \right)}{\log\frac{1}{\gamma}},
    \end{align*}
    and that the sample size $M$ satisfies
    \begin{align*}
        \log M & \gtrsim \max\bigg\{ \frac{B^\frac{4d+8}{d+2}}{C_d M^{\frac{d}{d+2}}} \log\frac{\tau T}{\delta}, \bigg(1 + \log\frac{R_d^*B}{C_d}\bigg)^\frac{2}{d}\bigg\},
    \end{align*}
    where $R_d^* \leq \max\{8dL, 2B+4dL\}$.
    Then with probability at least $1-\delta$, \cref{algo:iteration} with $T$ iterations on $\scrC_{B,L}$ gives mean-field term $z^T$ such that
    \begin{align*}
        & \|z^T-z^\star\|_1 \lesssim \frac{C_d L_F D^{5/2}}{(1-\gamma)(1-\kappa)\zeta} M^{-\frac{2}{d+2}}\log M + \frac{L_F D^{5/2}\epsilon}{(1-\gamma)(1-\kappa)\zeta} + \kappa^T Z + \frac{\kappa^{T-1}}{1-\gamma}. 
    \end{align*}
\end{theorem}

Theorem~\ref{thm:main-convex} suggests that
when we take convex regression on $\scrF = \scrC_{B,L}$ for the estimation of optimal value function in \cref{algo:q-iteration}, the estimation $z^T$ of \texttt{ConcaveHAM} converges to $z^\star$ in the rate of $\widetilde O(M^{-\frac{2}{d+2}})$ when the sample size $M$ is large enough. 
Note that $\kappa^T Z + \frac{\kappa^{T-1}}{1-\gamma}$ converges to zero exponentially fast, and the number of required iterations $T$ is in logarithmic order.
It is noteworthy that the number of iterations $\tau$ in \cref{algo:q-iteration} is kept in the order of $\log(M^\frac{2}{d+2}/\log M)$ to guarantee the convergence of \texttt{ConcaveHAM}. 

Due to the large covering entropy of $\scrC_{B,L}$ in the order of $d^d$, the sample complexity with the term $C_d$ in \cref{thm:main-convex} suffers from an almost linear dependence on $d$. 
This issue can be avoided by taking the regression solver $\texttt{LSE}(\scrA_{B,L}^{K,+},\calD_z,\calT_z Q;\epsilon)$ on the max-affine function set $\scrA_{B,L}^{K,+}$ with a specially designed $K$.

\begin{theorem}[Convergence for $\scrN_{B,L}^K$]\label{thm:main-affine-nn}
    Under Assumptions \ref{assump:convex}, \ref{assump:stochastic_concave}, \ref{assum:lipschitz-transition}, \ref{assum:lipschitz-mappings}, and \ref{assum:concentrability}, further assume that $\kappa<1$.
    For $\delta\in(0,1)$, suppose that the number of iterations within each \texttt{CFQI} subroutine satisfies 
    \begin{align*}
        \tau \asymp \frac{\log \frac{M^{4/(d+4)}}{d^2L^2\sqrt{D}}}{\log\frac{1}{\gamma}},
    \end{align*}
    and that the sample size $M$ satisfies
    \begin{align*}
        M \gtrsim \max\bigg\{&\frac{\log \frac{M^{4/(d+4)}}{d^2L^2\sqrt{D}}}{\delta\log\frac{1}{\gamma}}, \frac{\tau T}{\exp(2(d+1)M^{d/(d+4)})\delta},\bigg(\frac{B^3(B+dL)}{d\log M}\bigg)^{\frac{d+4}{2d+4}},
        \bigg(\frac{B^2}{d^2L^2}\log \frac{\tau T}{\delta}\bigg)^{\frac{d+4}{d}}\bigg\}.
    \end{align*}
    then with probability at least $1-\delta$, \cref{algo:iteration} with $T$ iterations on $\scrA_{B,L}^{K,+}$ or $\scrN_{B,L}^K$ gives mean-field term $z^T$ such that
    \begin{align*}
        & \|z^T-z^\star\|_1 \lesssim \frac{C_d' L_F D^{5/2}}{(1-\gamma)(1-\kappa)\zeta} M^{-\frac{4}{d+4}}\log M + \frac{L_F D^{5/2}\epsilon}{(1-\gamma)(1-\kappa)\zeta} + \kappa^T Z + \frac{\kappa^{T-1}}{1-\gamma}.
    \end{align*}
\end{theorem}

With the regression solver on max-affine function set $\scrA_{B,L}^K$, \texttt{ConcaveHAM} converges in the rate of $\widetilde O(M^{-\frac{4}{d+4}})$ when the sample size $M$ is large enough. 
Compared to \cref{thm:main-convex}, it not only enjoys a better rate but also enjoys a better constant in $d$, compared to the one in \cref{thm:main-convex} for $\scrC_{B,L}$.
\texttt{ConcaveHAM} also enjoys the same rate of $\widetilde O(M^{-\frac{4}{d+4}})$ for large $M$ when paired with regression on the set $\scrN_{B,L}^K$ of \ICNNs, and this is a natural statement following \cref{lem:exact-representation} and assuming the learning algorithms $\texttt{LSE}(\scrN_{B,L}^K,\calD_z,\calT_z Q;\epsilon)$ approach the global optimum on minimizing the empirical risk.

Therefore, the \texttt{ConvHAM} algorithm is provably convergent to the \QRMFE of the mean-field game at a sublinear rate, which provides a data-driven approach to a large class of heterogeneous agent models that does not require prior knowledge of the income shock transitions.
By taking advantage of convex regression, it is capable of capturing the economic intuitions, and at the same time, providing a continuous solution on value functions and avoiding the curse of dimensionality in discretization.

\section{Numerical Experiments}

\begin{figure*}[ht]
    \centering
    \includegraphics[width=0.45\linewidth]{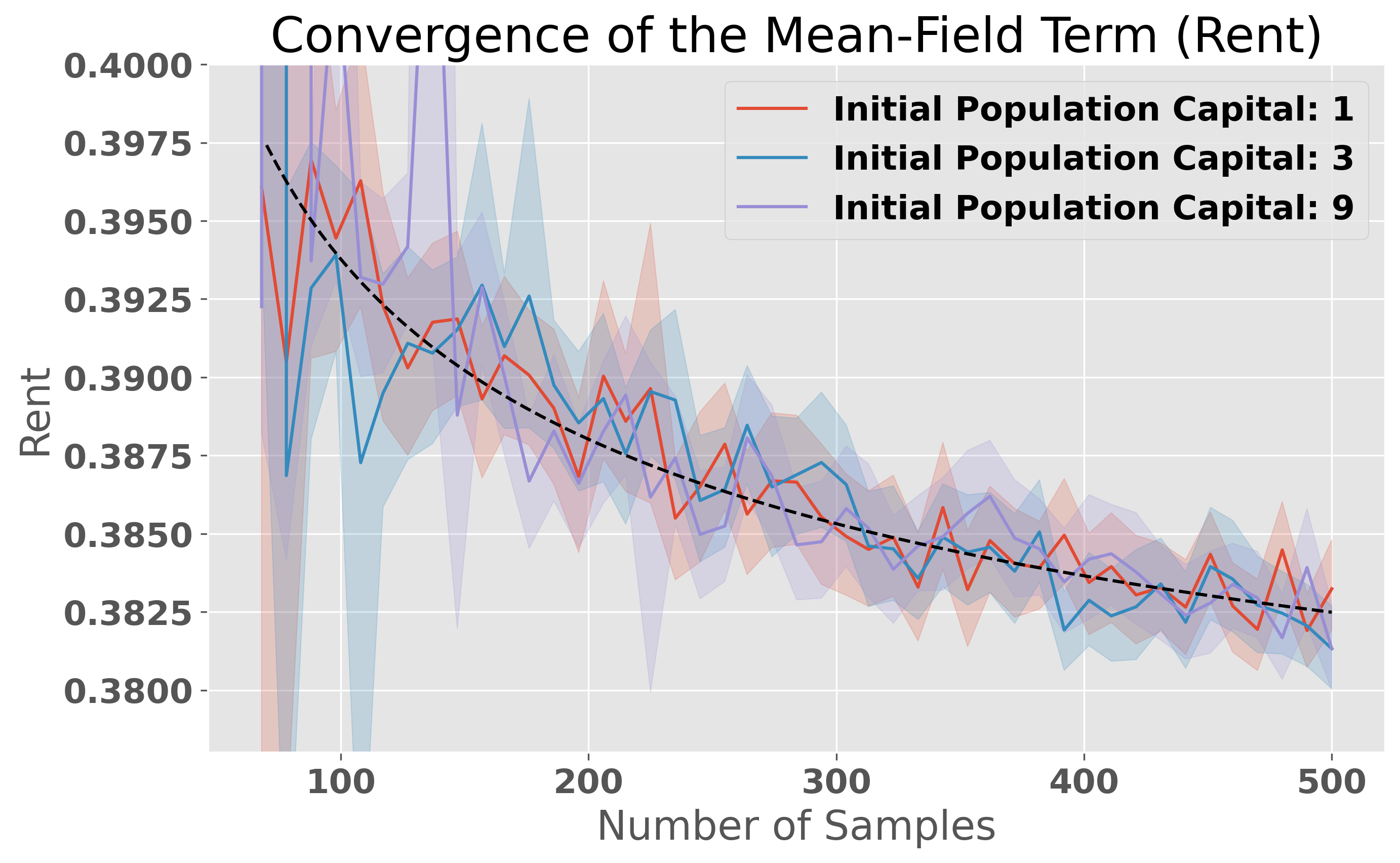}
    \includegraphics[width=0.45\linewidth]{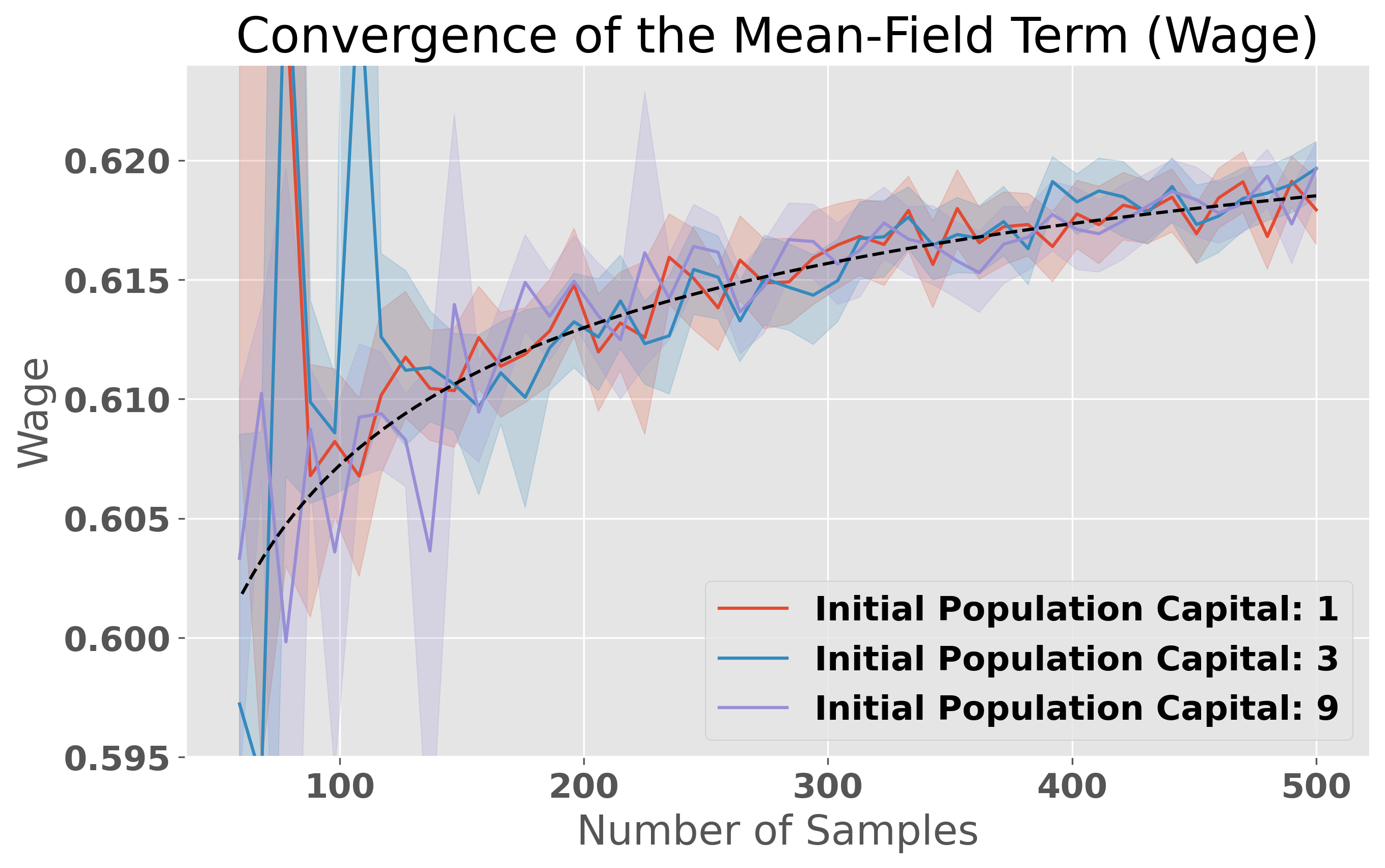}
    \caption{Convergence of the mean-field terms: rent $\varsigma$ and wage $\omega$. The red, blue, and purple lines show the average outputs of the mean-field terms by \texttt{ConvHAM} after 15 rounds of iterations with different number of samples $M$ provided at each iteration. The corresponding confidence bands are plotted over 50 random trials. The black dotted line indicates the convergence rate of order $\calO(M^{-2/5})$.}
    \label{fig}
\end{figure*}

To corroborate our theoretical results, we apply the proposed \texttt{ConvHAM} algorithm to some of the existing heterogeneous agent models in macroeconomics and show the convergence rate on mean-field terms that matches our theoretical guarantees.

\paragraph{Experiment setup.}
Our simulation features a canonical setup of the Aiyagari model where market incompleteness and idiosyncratic income shocks give rise to agent heterogeneity. 

Under this setup, the household has two levels of labor opportunities $n\in \{0, 1\}$, \ie, employment if $n=1$ and unemployment if $n=0$.
The utility function is logarithmic with respect to the amount of spending $\chi$ that is restricted by the household budget constraint $\chi \leq (1+\varsigma-\delta) b + \omega n - a$, where the mean-field term $z = (\omega, \varsigma)$ represents the wage and the rent on the market while $\delta$ denotes a depreciation constant. We restrict feasible capital holding $0\leq b \leq 20$ and the representative firm takes the Cobb-Douglas production function $F(b, n) = b^\alpha n^{1-\alpha}$ for any aggregated population capital $b$ and population labor $n$. The training data are sampled from a uniform distribution on feasible states $s=(b,n)$ following a uniform policy $\pi_0$.
 
\paragraph{Simulation results.}
Our method is able to achieve stable convergence in a few dozens of rounds of \texttt{ConvHAM} with moderate regularization, and the convergence rate matches our theoretical guarantees as illustrated in Figure~\ref{fig}.

More specifically, we aggregate the numerical results  over $50$ random trials for three initialization of population capital, which is determined by the initial mean-field terms, with a range of sample sizes $M$ for the concave value function estimation. 
We plot the mean-field terms (wage and rent) against different sample sizes for 15 rounds of \texttt{ConcaveHAM}, and compare them with the theoretically predicted convergence rate of $M^{-2/5}$ in terms of the sample size.
The numerical results verify that 
(i)~our algorithm achieves stable convergence, and
(ii) the convergence rate matches our theoretical guarantees.

\section{Related Work}\label{sec:related}

Heterogeneous agent models are a fundamental class of macroeconomic models that aim to understand the dynamics of individual participants of the market and how their population behaviors impact the overall performance of the economy \citep{hommes2006heterogeneous}. These models often consider random economic shocks to individual participants, such as aggregate and exogenous risk \citep{fernandez2019financial}, and focus on understanding the \RCE between households and firms \citep{kuhn2013recursive,light2020uniqueness}.
Machine learning has gained momentum in analyzing a range of economic models \citep{lepetyuk2020us,fernandez2020solving}.
Especially, deep \RL has claimed success in solving heterogeneous agent models with aggregate shocks and discrete-continuous choice dynamic models with heterogeneous agents \citep{han2021deepham,maliar2022deep}; it is also proven effective in finding micro-founded general equilibria for macroeconomic models with many agents \citep{curry2022finding}.

\MFG characterizes the decision-making dynamics of infinitely many players through a mean-field term, and it provides a natural extension to multi-agent games that admits simpler asymptotic analysis \citep{lasry2007mean,gueant2011mean}. 
Convergence of \RL algorithms to equilibria of \MFGs have been analyzed in recent works \citep{guo2019learning,subramanian2019reinforcement}. 
Particularly, \citet{guo2022entropy} studied \MFGs with entropy regularization on discrete action space and finite time horizon. 

Shape-constrained non-parametric modeling is a powerful tool that enables reliable estimation of special function classes, and it is especially useful for economic models that often demand shape constraints on utility functions \citep{blundell2003nonparametric}. 
In particular, theoretical guarantees have been laid out for convex and max-affine shape constraint regression in a thread of literature \citep{lim2012consistency,han2016multivariate,mazumder2019computational,seijo2011nonparametric,balazs2015near}.
An architecture of \ICNN has also been considered for the representation of convex functions with neural networks~\citep{amos2017input,chen2018optimal}.

\section{Conclusion}

We study a generalized Aiyagari model as an \MFG with concave shape constraint imposed on the agent's value function.
We propose a data-driven \RL framework \texttt{ConcaveHAM} for solving the \MFG that dramatically improves the computational efficiency with the help of a variety of convex regression procedures. The algorithm converges to the regularized recursive competitive equilibrium of the macroeconomic model at a sub-linear rate and can readily incorporate other forms of shape constraints.
To the best of our knowledge, this is the first data-driven algorithm for heterogeneous agent models with theoretical guarantees.

\subsubsection*{Acknowledgements}

Zhaoran Wang acknowledges National Science Foundation (Awards 2048075, 2008827, 2015568, 1934931), Simons Institute (Theory of Reinforcement Learning), Amazon, J.P. Morgan, and Two Sigma for their support.
Zhuoran Yang acknowledges Simons Institute (Theory of Reinforcement Learning) for the support.

\bibliographystyle{ims}
\bibliography{references}

\clearpage
\appendix

\onecolumn

\section{Clarification of Notation}\label{sec:notation}

We make a few clarifications on the problem setup of our heterogeneous agent model as well as the related notation that we intentionally omitted in the main text due to space constraints. 

\subsection{Bellman Operators}

We define Bellman operator $\calT_z^\pi$ for any policy $\pi$ and mean-field term $z$ that maps from any Q-function $Q$ to
\begin{align*}
     (\calT_z^\pi Q)(s,a) & \defeq r_z(s,a) + \gamma \expect[\inner{Q(s',\cdot)}{\pi(\cdot \given s')}_{\Gamma(s')}\given s,a] \\
     & = r_z(s,a) + \gamma (\calP_z V_z^\pi)(s,a),
\end{align*}
where $(\calP_z v)(s,a) \defeq \expect[v(s')\given s,a]$ and $s'\sim\PP_z(\cdot\given s,a)$ be the next state for any state-value function $v$. 
Further, we also denote $\calT_z$ to be the Bellman optimality operator, such that for any value function $Q$
\begin{align*}
    (\calT_z Q)(s,a) \defeq r_z(s,a) + \gamma \expect\Big[\max_{a'\in\Gamma(s')} Q(s',a')\ \Big|\ s,a\Big]
\end{align*}
under a greedy policy with respect to $Q$.
Note that the value function $Q_z^\pi$ corresponding to any policy $\pi$ is the fixed point of $\calT_z^\pi$, \ie, $Q_z^\pi(s,a) = (\calT_z^\pi Q_z^\pi)(s,a)$, and for any mean-field term $z$, the optimal Q-function $Q_z^\star(s)$ is the fixed point of the Bellman optimality operator $\calT_z$, \ie,
\begin{align*}
    Q_z^\star(s,a) = (\calT_z Q_z^\star)(s,a) \defeq r_z(s,a) + \gamma \expect\Big[\max_{a'\in\Gamma(s')} Q_z^\star(s',a')\ \Big|\ s,a\Big]
\end{align*} 
and the optimal state-value function follows from $V_z^\star(s) = \max_{a\in\Gamma(s)} Q_z^\star(s,a)$.

\subsection{Issues With \RCE}\label{sec:issues-rce}

For any concave value function $Q$, there may exist multiple actions maximizing $Q(s,\cdot)$ for any $s\in\calS$, and this non-uniqueness in optimal policy leads to the non-uniqueness of the aggregate indicator $\Xi^\star$ and therefore the non-uniqueness of \RCE. 

The non-robustness of an equilibrium implies that a small estimation error on the optimal value function can lead to a large estimation error on the corresponding optimal policy.

\subsection{Regularized \RCE}

Through adding an entropy regularization, we have a unique regularized optimal policy given any mean-field term, and for the convenience of notation, we denote the regularized optimal policy for any mean-field term $z$ as
\begin{align*}
    \Lambda_1(y) \defeq \bigg\{ \pi^\dagger \st \pi^\dagger = \argmax_{\pi} \Big\{\int_{\Gamma(s)} \pi(a\given s)Q_z^\star(s,a) \diff a - \calH(\pi(\cdot\given s))\Big\} \text{ for all } s\in\calS \bigg\}.
\end{align*}
For any policy $\pi$, we also denote the next-iteration mean-field term under any current mean-field term $z$ with a shorthand
\begin{align*}
    \Lambda_2(\pi,z) \defeq \Phi(\Psi(z,\pi)).
\end{align*}
To conclude the existence and uniqueness of the \QRMFE, we need to show the convergence of the above iteration process starting with any mean-field term $z\in\calZ$, and it suffices if $\Lambda_2(\Lambda_1(\cdot))$ is a contractive mapping.

\section{Proof of Main Results}

\subsection{Proof of \cref{thm:contraction}}\label{sec:contraction-proof}

\begin{proof}
    Recall that the optimal regularized policy $\pi_z^\star$ of the representative agent under any mean-field term $z$ is defined as
    \begin{align*}
        \pi_z^\star = \Lambda_1(y),
    \end{align*}
    and the unique mean-field term induced by the best response of the environment to $\pi_z^\star$ under any mean-field term $z$ is given by
    \begin{align*}
        \Lambda_2(\pi_z^\star,z) \defeq \Phi(\Psi(\pi_z^\star,\pi)).
    \end{align*}
    Let us further denote $\Lambda(y) \defeq \Lambda_2(\Lambda_1(y),z)$ to be the operator that maps any mean-field term $z$ to its corresponding next-iteration mean-field term through the regularized policy $\Lambda_1(y)$. A policy-environment pair $(\pi^\star,z^\star)$ is a \QRMFE if and only if $z^\star = \Lambda(z^\star)$ and $\pi^\star = \Lambda_1(z^\star)$. 
    Similar to the argument in \citet{guo2020general}, note that for any mean-field terms $z$ and $z'$, the distance between their corresponding next-iteration mean-field terms is controlled by
    \begin{align*}
        \|\Lambda(y) - \Lambda(z')\|_1 & = \|\Lambda_2(\Lambda_1(y),z) - \Lambda_2(\Lambda_1(z'),z')\|_1 \\
        & \leq \|\Lambda_2(\Lambda_1(y),z) - \Lambda_2(\Lambda_1(z'),z)\|_1 + \|\Lambda_2(\Lambda_1(z'),z) - \Lambda_2(\Lambda_1(z'),z')\|_1 \\
        & \leq \|\Phi(\Psi(\pi_z^\star,z)) - \Phi(\Psi(\pi_{z'}^\star,z))\|_1 + \|\Phi(\Psi(\pi_{z'}^\star,z)) - \Phi(\Psi(\pi_{z'}^\star,z'))\|_1,
        \end{align*}
        where the first equality follows from the definition of operator $\Lambda(\cdot)$, the first inequality follows from triangular inequality, and the last inequality follows from the definition of $\Lambda_1(\cdot)$ and $\Lambda_2(\cdot,\cdot)$. 
        Following \cref{assum:lipschitz-mappings}, the Lipschitzness of the mappings $\Psi$ and $\Phi$ yields
        \begin{align*}
            \|\Phi(\Psi(\pi_{z'}^\star,z)) - \Phi(\Psi(\pi_{z'}^\star,z'))\|_1 \leq L_F \|z-z'\|_1
        \end{align*}
        and
        \begin{align*}
            \|\Phi(\Psi(\pi_z^\star,z)) - \Phi(\Psi(\pi_{z'}^\star,z))\|_1 & \leq L_F\cdot \|\pi_z^\star - \pi_{z'}^\star\|_{\infty,\overline\nu}.
        \end{align*}
        Further, it also holds for the optimal Q-functions $Q_z^\star$ and $Q_{z'}^\star$ that the distance between their corresponding optimal policies under entropy regularization $\calH$ is bounded through the distance between the mean-field terms, \ie,
        \begin{align*}
            \|\pi_z^\star - \pi_{z'}^\star\|_{\infty,\overline\nu} & \leq D\cdot \|\pi_z^\star - \pi_{z'}^\star\|_{\infty,\widetilde\nu} \\
            & \leq \frac{D^2}{\zeta} \|Q_z^\star-Q_{z'}^\star\|_{1,\widetilde\mu} \\
            & \leq \frac{J D^2}{\zeta} \|z-z'\|_1,
        \end{align*}
    where the first inequality is due to the concentrability of the sampling distribution $\widetilde\nu$ in \cref{assum:concentrability}, the second inequality follows by applying \cref{lem:lipschitz-policy}, and the last inequality is due to \cref{lem:lipschitz-action-value}.
    Hence, we combine the above inequalities together to get 
    \begin{align*}
        \|\Lambda(y) - \Lambda(z')\|_1 & \leq L_F \|z-z'\|_1 + L_F\cdot \|\pi_z^\star - \pi_{z'}^\star\|_{\infty,\overline\nu} \\
        & \leq \Big(\frac{J L_F D^2}{\zeta} + L_F\Big) \|z-z'\|_1,
    \end{align*}
    and the operator $\Lambda(\cdot)$ forms a contraction if $\kappa \defeq \frac{J L_F D^2}{\zeta} + L_F < 1$. The existence and uniqueness of the \QRMFE $(\pi^\star,z^\star)$ follows from Banach fixed-point theorem on compact metric space $\calZ$ and the uniqueness of global optimizer on strongly convex optimization.
\end{proof}

\begin{lemma}\label{lem:lipschitz-action-value}
    Let $z,z'\in\calZ$ be two mean-field terms and $Q_z^\star$ and $Q_{z'}^\star$ to be the optimal value functions under the \MDPs with transition kernels $\PP_z$ and $\PP_{z'}$, respectively. Define $B \defeq 1/(1-\gamma)$. Then under Assumptions \cref{assum:lipschitz-transition}, and \ref{assum:concentrability}, we have
    \begin{align*}
        \|Q_z^\star - Q_{z'}^\star\|_{1,\widetilde\mu} \leq J \cdot \|z - z'\|_1,
    \end{align*}
    where 
    \begin{align*}
        J \defeq \frac{L_r+\gamma B L_\PP}{1-\gamma D}.
    \end{align*}
\end{lemma}

\begin{proof}
    Notice that for any optimal Q-functions $Q_z^\star$ and $Q_{z'}^\star$, it follows from the Bellman optimality equation that
    \begin{align*}
        \|Q_z^\star - Q_{z'}^\star\|_{1,\widetilde\mu} & = \|\calT_z Q_z^\star - \calT_{z'} Q_{z'}^\star\|_{1,\widetilde\mu} \\
        & = \Big\| r_z(s,a) + \gamma \int_\calS \PP_z(s'\given s,a) \max_{a\in \Gamma(s')} Q_z^\star(s',a) \diff s' \\
        & \qquad - r_{z'}(s,a) - \gamma \int_\calS \PP_{z'}(s'\given s,a) \max_{a\in \Gamma(s')} Q_{z'}^\star(s',a) \diff s' \Big\|_{1,\widetilde\mu},
    \end{align*}
    where the first equality is due to the fact that optimal Q-functions $Q_z^\star$ and $Q_{z'}^\star$ are the fixed points under the Bellman operators $\calT_z$ and $\calT_{z'}$ respectively, the second equality expands the Bellman operators following the definition of $\calT_z$ and $\calT_{z'}$.
    Further, we decompose the difference into three pieces, \ie,
    \begin{align}\label{eqn:bellman-split}
        \begin{split}
            \|Q_z^\star - Q_{z'}^\star\|_{1,\widetilde\mu} & \leq \gamma \Big\| \int_\calS \PP_z(s'\given s,a) \big[\max_{a\in \Gamma(s')}\{Q_z^\star(s',a)\} - \max_{a\in \Gamma(s')}\{Q_{z'}^\star(s',a)\}\big] \diff s' \Big\|_{1,\widetilde\mu} \\
            & \qquad + \gamma \Big\| \int_\calS \big[\PP_z(s'\given s,a) - \PP_{z'}(s'\given s,a)\big] \max_{a\in \Gamma(s')} Q_{z'}^\star(s',a) \diff s' \Big\|_{1,\widetilde\mu} \\
            & \qquad + \|r_z(s,a) - r_{z'}(s,a)\|_{1,\widetilde\mu},
        \end{split}
    \end{align}
    where the inequality follows from triangular inequality and supplying an intermediate dummy term $\int_\calS \PP_z(s'\given s,a)\max_{a\in \Gamma(s')} Q_{z'}^\star(s',a) \diff s'$. Notice that the last term in \eqref{eqn:bellman-split} can be bounded following the Lipschitzness of the reward function, \ie,
    \begin{align*}
        \|r_z(s,a) - r_{z'}(s,a)\|_{1,\widetilde\mu} \leq L_r \cdot \|z - z'\|_1,
    \end{align*}
    and the second term in \eqref{eqn:bellman-split} is controlled by the distance between transition kernels $\PP_z$ and $\PP_{z'}$, \ie,
    \begin{align*}
        \Big\| \int_\calS \big[\PP_z(s'\given s,a) - \PP_{z'}(s'\given s,a)\big] \max_{a\in \Gamma(s')} Q_{z'}^\star(s',a) \diff s' \Big\|_{1,\widetilde\mu} & \leq L_\PP\cdot \|Q^\star\|_{\infty} \|z - z'\|_1 \\
        & \leq B L_\PP \|z - z'\|_1,
    \end{align*}
    where the first inequality is due to the Lipschitz condition on $\PP_z$ with respect to $z$ in \cref{assum:lipschitz-transition} and $\max_{a\in \Gamma(s')} Q_{z'}^\star(s',a) \leq \|Q^\star\|_\infty \leq B$ for all feasible state-action pair $(s,a)$.
    The second inequality follows from the assumption that the reward $r_h(s,a)\in [0,1]$ for any feasible state-action pair $(s,a)$ and step $h\in\ZZ_+$. 
    Finally, for any optimal Q-functions $Q_z^\star$ and $Q_{z'}^\star$, the first term in \eqref{eqn:bellman-split} is bounded by
    \begin{align*}
        & \Big\| \int_\calS \PP_z(s'\given s,a) \big[\max_{a'\in \Gamma(s')}\{Q_z^\star(s',a')\} - \max_{a'\in \Gamma(s')}\{Q_{z'}^\star(s',a')\}\big] \diff s' \Big\|_{1,\widetilde\mu} \\
        & \qquad \leq \Big\| \int_\calS \PP_z(s'\given s,a) \max_{a'\in \Gamma(s')}\big|Q_z^\star(s',a') - Q_{z'}^\star(s',a')\big| \diff s' \Big\|_{1,\widetilde\mu} \\
        & \qquad = \big\| \max_{a\in \Gamma(s)}\big|Q_z^\star(s,a) - Q_{z'}^\star(s,a)\big| \big\|_{1,\PP_z\widetilde\mu} \\
        & \qquad = \big\| Q_z^\star(s,a) - Q_{z'}^\star(s,a) \big\|_{1,\PP_z\widetilde\mu\times\pi^\ddagger(Q_z^\star,Q_{z'}^\star)}
    \end{align*}
    where the inequality follows from the fact that $|\max_{x\in\calX} f(x) - \max_{y\in\calX} g(y)| \leq \max_{x\in\calX} |f(x) - g(x)|$ for any $x,y$ in the shared domain $\calX$, the first equality is due to the shorthand $s'\sim\PP_z\mu$ to denote $s'\sim \PP_z(\cdot\given s,a)$ and $(s,a)\sim \mu$, and the last equality is due to defining the policy $\pi^\ddagger(Q_z^\star,Q_{z'}^\star)$ that takes the action to maximize the difference between Q-functions $Q_z^\star$ and $Q_{z'}^\star$, \ie, $\pi^\ddagger(Q_z^\star,Q_{z'}^\star)(s) \defeq \argmax_{a\in\calA} |Q_z^\star(s,a) - Q_{z'}^\star(s,a)|$.
    Putting the bounds above together, we have an upper bound
    \begin{align}\label{eqn:Q-diff-bound}
        \|Q_z^\star - Q_{z'}^\star\|_{1,\widetilde\mu} & \leq \gamma D \cdot \|Q_z^\star - Q_{z'}^\star\|_{1,\widetilde\mu} + \gamma B L_\PP \cdot \|z - z'\|_1 + L_r \cdot \|z - z'\|_1,
    \end{align}
    where the inequality follows from \cref{assum:concentrability}. Rearranging the term in \eqref{eqn:Q-diff-bound} and following the definition of $J$ to get
    \begin{align*}
        \|Q_z^\star - Q_{z'}^\star\|_{1,\widetilde\mu} & \leq \frac{L_r+\gamma B L_\PP}{1-\gamma D} \|z - z'\|_1 = J\cdot \|z - z'\|_1.
    \end{align*}
\end{proof}

\begin{lemma}\label{lem:lipschitz-policy}
    Let $Q_z^\star$ and $Q_{z'}^\star$ be two optimal value functions under the \MDPs with transition kernels $\PP_z$ and $\PP_{z'}$, respectively. Then, under \cref{assum:concentrability}, we have
    \begin{align*}
        \|\pi_z^\star(\cdot \given \cdot) - \pi_{z'}^\star(\cdot \given \cdot)\|_{\infty,\widetilde\nu} \leq \frac{D}{\zeta} \cdot \|Q_z^\star-Q_{z'}^\star\|_{1,\widetilde\mu}.
    \end{align*}
\end{lemma}

\begin{proof}
    For any mean-field term $z$ and its corresponding optimal value function $Q_z^\star$, the optimal regularized policy $\pi_z^\star$ can be written as
    \begin{align*}
        \pi_z^\star(\cdot \given s) & = \argmax_{u \in \Delta_{\Gamma(s)}} \left\{\int_{\Gamma(s)} u(a)Q_z^\star(s,a) \diff a - \calH(u)\right\} \\
        & = \argmax_{u \in \Delta_{\Gamma(s)}}\{\langle u, Q_z^\star(s,\cdot)\rangle-\calH(u)\} \\
        & = \nabla \calH^\star(Q(s,\cdot)),
    \end{align*}
    where the first equality follows from the definition of the optimal regularized policy $\pi_z^\star$, and the last equality follows from \cref{lem:dual} with $\nabla \calH^\star(Q(s,\cdot))$ being the functional derivative of the dual functional $\calH^*$ evaluated at $Q(s,\cdot)$ for $s\in\calS$. Similarly, we write the optimal regularized policy $\pi_{z'}^\star(\cdot \given s) = \nabla \calH^\star(Q_{z'}^\star(s,\cdot))$ for any $s\in\calS$, and it follows that the distance between optimal regularized policies is controlled by that between optimal value functions, \ie,
    \begin{align*}
        \|\pi_z^\star(\cdot \given \cdot) - \pi_{z'}^\star(\cdot \given \cdot)\|_{\infty,\widetilde\nu} & = \|\nabla \calH^\star(Q_z^\star(s,\cdot)) - \nabla \calH^\star(Q_{z'}^\star(s,\cdot))\|_{\infty,\widetilde\nu} \\
        & \leq \frac{1}{\zeta} \|Q_z^\star(s,\cdot) - Q_{z'}^\star(s,\cdot)\|_{1,\widetilde\nu} \\
        & \leq \frac{D}{\zeta} \|Q_z^\star - Q_{z'}^\star\|_{1,\widetilde\mu},
    \end{align*}
    where the first inequality is due to the $\zeta$-strong convexity assumption on $\calH$ and \cref{lem:dual}, and the last inequality is due to \cref{assum:concentrability} and $\|Q_z^\star(s,\cdot) - Q_{z'}^\star(s,\cdot)\|_{1,\widetilde\nu} = \|Q_z^\star - Q_{z'}^\star\|_{1,\widetilde\nu\times \pi_U}$ with $\pi_U$ being a uniform policy.
\end{proof}

From the proof of Lemma~\ref{lem:lipschitz-policy}, we bound the distance between the optimal household policies under different mean-field environments.
It provides an abstraction over multi-agent games where exponentially many interactions between the agents are considered. 
Recall that mean-field setups, viewed as a limit of their multi-agent counterparts~\citep{song2021can,dubey2021provably,he2022simple,xu2021meta,wang2023breaking,ling2019landscape}, drive the number of agents to infinity. 
Compared to the multi-agent setting, in our analysis we bound the distance between the optimal household policies under different mean-field environments.

\subsection{Proof of \cref{thm:main-convex}}\label{sec:main-convex-proof}

\begin{proof}
    For each $t\in[T]$, the mean-field term $z^t$ is determined by the mean-field term $z^{t-1}$ from last iteration of \cref{algo:iteration} and the regularized policy $\pi_{\widehat Q^{t-1}}$ induced by the estimator $\widehat Q^{t-1}$ of the optimal value function $Q_{z^{t-1}}^\star$. Hence, we follow the definition of operator $\Lambda_1$ and $\Lambda_2$, and the distance between $z^t$ and the regularized \RCE $z^\star$ is bounded by
    \begin{align*}
        \|z^t-z^\star\|_1 & = \|\Lambda_2(\pi_{\widehat Q^{t-1}},z^{t-1}) - \Lambda_2(\Lambda_1(z^\star),z^\star)\|_1 \\
        & \leq \|\Lambda_2(\Lambda_1(z^{t-1}),z^{t-1}) - \Lambda_2(\Lambda_1(z^\star),z^\star)\|_1 \\
        & \qquad + \|\Lambda_2(\Lambda_1(z^{t-1}),z^{t-1}) - \Lambda_2(\pi_{\widehat Q^{t-1}},z^{t-1})\|_1 \\
        & \leq \Big(\frac{J L_F D^2}{\zeta} + L_F\Big) \|z^{t-1}-z^\star\|_1 + L_F \|\pi_{Q_{z^{t-1}}^\star} - \pi_{\widehat Q^{t-1}}\|_{\infty,\overline\nu},
    \end{align*}
    where the equality follows from the definition of $\Lambda_1$ and $\Lambda_2$ as well as \cref{def:ne}, the first inequality follows from triangular inequality, and the second inequality is due to \cref{thm:contraction} and \cref{assum:lipschitz-mappings}. We denote $\kappa \defeq J L_F D^2/\zeta + L_F$ for the convenience of notation. Following \cref{lem:lipschitz-policy}, we further have
    \begin{align}
        \|z^t-z^\star\|_1 & \leq \Big(\frac{J L_F D^2}{\zeta} + L_F\Big) \|z^{t-1}-z^\star\|_1 + \frac{L_F D^2}{\zeta} \|\widehat Q^{t-1} - Q_{z^{t-1}}^\star\|_{1,\widetilde\mu} \nonumber \\
        & \leq \kappa\cdot \|z^{t-1}-z^\star\|_1 + \frac{L_F D^2}{\zeta} \|\widehat Q^{t-1} - Q_{z^{t-1}}^\star\|_{\widetilde\mu}. \label{eqn:recursive}
    \end{align}
    After expanding the the first term of \eqref{eqn:recursive} recursively for $t=T$, it holds that for any $T\in\ZZ_+$, the distance between the output mean-field term $z^T$ of \cref{algo:iteration} and $z^\star$ is controlled by
    \begin{align*}
        \|z^T-z^\star\|_1 & \leq \kappa^T \|z^0-z^\star\|_1 + \frac{L_F D^2}{\zeta} \sum_{t=1}^{T-1} \kappa^{T-1-t} \|\widehat Q^t - Q_{z^t}^\star\|_{\widetilde\mu} + \kappa^{T-1} \|\widehat Q^0 - Q_{p^0}^\star\|_{\widetilde\mu} \\
        & \leq \kappa^T \|z^0-z^\star\|_1 + \frac{L_F D^2}{\zeta} \sup_{t\in[T]} \|\widehat Q^t - Q_{z^t}^\star\|_{\widetilde\mu} \sum_{t=1}^{T-1} \kappa^{T-1-t} + \kappa^{T-1} \|\widehat Q^0 - Q_{p^0}^\star\|_{\widetilde\mu} \\
        & \leq \kappa^T \|z^0-z^\star\|_1 + \frac{L_F D^2}{(1-\kappa)\zeta} \sup_{t\in[T]} \|\widehat Q^t - Q_{z^t}^\star\|_{\widetilde\mu} + \frac{\kappa^{T-1}}{1-\gamma},
    \end{align*}
    where the first inequality follows from recursive expansion of \eqref{eqn:recursive}, the second inequality is due to uniform bounding of the estimation error $\|\widehat Q^t - Q_{z^t}^\star\|_{\widetilde\mu} \leq \sup_{t\in[T]} \|\widehat Q^t - Q_{z^t}^\star\|_{\widetilde\mu}$ for all $t\in[T]$, and the last inequality is due to the bounded optimal value function $\|Q_z^*\|_\infty \leq \frac{1}{1-\gamma} = B$ for all $z\in\calZ$.
    Therefore, for regression on the set $\scrC_{B,L}$ of concave functions, we have with probability $1-\delta$ that
    \begin{align*}
        \|z^T-z^\star\|_1 & \lesssim \frac{L_F D^2\sqrt{D}}{(1-\gamma)(1-\kappa)\zeta} \left[B^\frac{d+8}{d+2} (d+1)^\frac{2}{d+2} (dL)^\frac{d}{d+2} M^{-\frac{2}{d+2}}\log M + \epsilon\right] + \kappa^T Z + \frac{\kappa^{T-1}}{1-\gamma}
    \end{align*}
    for sample size $M^\frac{d}{d+2} \gtrsim \log M$, $M^\frac{d}{d+2}\log M \gtrsim B^\frac{3d}{d+2} (d+1)^\frac{-2}{d+2} (dL)^\frac{-d}{d+2} \log\frac{\tau T}{\delta}$, and  
    \begin{align*}
        (\log M)^\frac{d}{2} \gtrsim 1 + \log(R_d^*B^\frac{-6}{d+2} (d+1)^\frac{-2}{d+2} (dL)^\frac{-d}{d+2}),
    \end{align*}
    where $R_d^* \leq \max\{8dL, 2B+4dL\}$ and
    \begin{align*}
        \tau \asymp \log\left( D^{-1/2} B^{-\frac{d+8}{d+2}}(d+1)^{-\frac{2}{d+2}}(dL)^{-\frac{d}{d+2}}M^\frac{2}{d+2}\frac{1}{\log M} \right) / \log\frac{1}{\gamma}.
    \end{align*}
    The inequality follows from applying \cref{lem:fqi-approx} with a union bound. 
\end{proof}

\subsection{Proof of \cref{thm:main-affine-nn}}\label{sec:main-affine-nn-proof}

Combine the statements of \cref{thm:main-affine,thm:main-nn} to get \cref{thm:main-affine-nn}.

\begin{theorem}\label{thm:main-affine}
    Let $B \defeq 1/(1-\gamma)$, $L \defeq L_r/(1-\gamma)$.
    Assume that $\kappa \defeq \frac{J L_F D^2}{\zeta} + L_F<1$, then with probability at least $1-\delta$, \cref{algo:iteration} with $T$ iterations on $\scrA_{B,L}^K$ gives mean-field term $z^T$ such that
    \begin{align*}
        \|z^T-z^\star\|_1 & \lesssim \frac{L_F D^2\sqrt{D}}{(1-\gamma)(1-\kappa)\zeta} \left[\left( (d+1)B^4 + (B+dL)dL \right) M^{-\frac{4}{d+4}}\log M + \epsilon\right] + \kappa^T Z + \frac{\kappa^{T-1}}{1-\gamma}
    \end{align*}
    for sample size $M \gtrsim \frac{1}{\delta\log\frac{1}{\gamma}} \log \frac{M^{4/(d+4)}}{d^2L^2\sqrt{D}}$, $M^\frac{2d+4}{d+4} \log M \gtrsim \frac{B^3(B+4dL)}{d+1}$, $M \gtrsim \frac{\tau T}{\exp(2(d+1)M^{d/(d+4)})\delta}$, and $M^\frac{d}{d+4} \gtrsim \frac{B^2}{d^2L^2}\log \frac{2\tau T}{\delta}$,
    where $\tau \asymp \log \frac{M^{4/(d+4)}}{d^2L^2\sqrt{D}} / \log\frac{1}{\gamma}$.
\end{theorem}

\begin{proof}
    Following similar argument as in the proof of \cref{thm:main-convex}, the distance between the output mean-field term $z^T$ of \cref{algo:iteration} and $z^\star$ is controlled by
    \begin{align*}
        \|z^T-z^\star\|_1 & \leq \kappa^T \|z^0-z^\star\|_1 + \frac{L_F D^2}{\zeta} \sum_{t=1}^{T-1} \kappa^{T-1-t} \|\widehat Q^t - Q_{z^t}^\star\|_{\widetilde\mu} + \kappa^{T-1} \|\widehat Q^0 - Q_{p^0}^\star\|_{\widetilde\mu} \\
        & \leq \kappa^T \|z^0-z^\star\|_1 + \frac{L_F D^2}{\zeta} \sup_{t\in[T]} \|\widehat Q^t - Q_{z^t}^\star\|_{\widetilde\mu} \sum_{t=1}^{T-1} \kappa^{T-1-t} + \kappa^{T-1} \|\widehat Q^0 - Q_{p^0}^\star\|_{\widetilde\mu} \\
        & \leq \kappa^T \|z^0-z^\star\|_1 + \frac{L_F D^2}{(1-\kappa)\zeta} \sup_{t\in[T]} \|\widehat Q^t - Q_{z^t}^\star\|_{\widetilde\mu} + \frac{\kappa^{T-1}}{1-\gamma},
    \end{align*}
    where the first inequality follows from recursive expansion of \eqref{eqn:recursive}, the second inequality is due to uniform bounding of the estimation error $\|\widehat Q^t - Q_{z^t}^\star\|_{\widetilde\mu} \leq \sup_{t\in[T]} \|\widehat Q^t - Q_{z^t}^\star\|_{\widetilde\mu}$ for all $t\in[T]$, and the last inequality is due to the bounded optimal value function $\|Q_z^*\|_\infty \leq B$ for all $z\in\calZ$.
    For regression on max-affine functions $\scrA_{B,L}^K$, we have with probability $1-\delta$ that
    \begin{align*}
        \|z^T-z^\star\|_1 & \lesssim \frac{L_F D^2\sqrt{D}}{(1-\gamma)(1-\kappa)\zeta} \left[\left( (d+1)B^4 + (B+dL)dL \right) M^{-\frac{4}{d+4}}\log M + \epsilon\right] + \kappa^T Z + \frac{\kappa^{T-1}}{1-\gamma}
    \end{align*}
    for
    \begin{gather*}
        M \gtrsim \frac{1}{\delta\log\frac{1}{\gamma}} \log \frac{M^{4/(d+4)}}{d^2L^2\sqrt{D}}, \qquad M^\frac{2d+4}{d+4} \log M \gtrsim \frac{B^3(B+4dL)}{d+1}, \\
        M \gtrsim \frac{\tau T}{\exp(2(d+1)M^{d/(d+4)})\delta}, \qquad M^\frac{d}{d+4} \gtrsim \frac{B^2}{d^2L^2}\log \frac{2\tau T}{\delta},
    \end{gather*} 
    where
    \begin{align*}
        \tau \asymp \log \frac{M^{4/(d+4)}}{d^2L^2\sqrt{D}} / \log\frac{1}{\gamma}.
    \end{align*}
    The inequality follows from applying \cref{lem:fqi-approx} with a union bound. 
\end{proof}

\begin{theorem}\label{thm:main-nn}
    Let $B \defeq 1/(1-\gamma)$, $L \defeq L_r/(1-\gamma)$.
    Assume that $\kappa \defeq \frac{J L_F D^2}{\zeta} + L_F<1$, then with probability at least $1-\delta$, \cref{algo:iteration} with $T$ iterations on $\scrN_{B,L}^K$ gives mean-field term $z^T$ such that
    \begin{align*}
        \|z^T-z^\star\|_1 & \lesssim \frac{L_F D^2\sqrt{D}}{(1-\gamma)(1-\kappa)\zeta} \left[\left( (d+1)B^4 + (B+dL)dL \right) M^{-\frac{4}{d+4}}\log M + \epsilon\right] + \kappa^T Z + \frac{\kappa^{T-1}}{1-\gamma}
    \end{align*}
    for sample size $M \gtrsim \frac{1}{\delta\log\frac{1}{\gamma}} \log \frac{M^{4/(d+4)}}{d^2L^2\sqrt{D}}$, $M^\frac{2d+4}{d+4} \log M \gtrsim \frac{B^3(B+4dL)}{d+1}$, $M \gtrsim \frac{\tau T}{\exp(2(d+1)M^{d/(d+4)})\delta}$, and $M^\frac{d}{d+4} \gtrsim \frac{B^2}{d^2L^2}\log \frac{2\tau T}{\delta}$,
    where $\tau \asymp \log \frac{M^{4/(d+4)}}{d^2L^2\sqrt{D}} / \log\frac{1}{\gamma}$.
\end{theorem}

\begin{proof}
    Following a similar argument as in the proof of \cref{thm:main-convex}, the distance between the output mean-field term $z^T$ of \cref{algo:iteration} and $z^\star$ is controlled by
    \begin{align*}
        \|z^T-z^\star\|_1 & \leq \kappa^T \|z^0-z^\star\|_1 + \frac{L_F D^2}{\zeta} \sum_{t=1}^{T-1} \kappa^{T-1-t} \|\widehat Q^t - Q_{z^t}^\star\|_{\widetilde\mu} + \kappa^{T-1} \|\widehat Q^0 - Q_{p^0}^\star\|_{\widetilde\mu} \\
        & \leq \kappa^T \|z^0-z^\star\|_1 + \frac{L_F D^2}{\zeta} \sup_{t\in[T]} \|\widehat Q^t - Q_{z^t}^\star\|_{\widetilde\mu} \sum_{t=1}^{T-1} \kappa^{T-1-t} + \kappa^{T-1} \|\widehat Q^0 - Q_{p^0}^\star\|_{\widetilde\mu} \\
        & \leq \kappa^T \|z^0-z^\star\|_1 + \frac{L_F D^2}{(1-\kappa)\zeta} \sup_{t\in[T]} \|\widehat Q^t - Q_{z^t}^\star\|_{\widetilde\mu} + \frac{\kappa^{T-1}}{1-\gamma},
    \end{align*}
    where the first inequality follows from recursive expansion of \eqref{eqn:recursive}, the second inequality is due to uniform bounding of the estimation error $\|\widehat Q^t - Q_{z^t}^\star\|_{\widetilde\mu} \leq \sup_{t\in[T]} \|\widehat Q^t - Q_{z^t}^\star\|_{\widetilde\mu}$ for all $t\in[T]$, and the last inequality is due to the bounded optimal value function $\|Q_z^*\|_\infty \leq B$ for all $z\in\calZ$.
    For regression on max-affine functions $\scrN_{B,L}^K$, we have with probability $1-\delta$ that
    \begin{align*}
        \|z^T-z^\star\|_1 & \lesssim \frac{L_F D^2\sqrt{D}}{(1-\gamma)(1-\kappa)\zeta} \left[\left( (d+1)B^4 + (B+dL)dL \right) M^{-\frac{4}{d+4}}\log M + \epsilon\right] + \kappa^T Z + \frac{\kappa^{T-1}}{1-\gamma}
    \end{align*}
    for
    \begin{gather*}
        M \gtrsim \frac{1}{\delta\log\frac{1}{\gamma}} \log \frac{M^{4/(d+4)}}{d^2L^2\sqrt{D}}, \qquad M^\frac{2d+4}{d+4} \log M \gtrsim \frac{B^3(B+4dL)}{d+1}, \\
        M \gtrsim \frac{\tau T}{\exp(2(d+1)M^{d/(d+4)})\delta}, \qquad M^\frac{d}{d+4} \gtrsim \frac{B^2}{d^2L^2}\log \frac{2\tau T}{\delta},
    \end{gather*} 
    where 
    \begin{align*}
        \tau \asymp \log \frac{M^{4/(d+4)}}{d^2L^2\sqrt{D}} / \log\frac{1}{\gamma}.
    \end{align*}
    The inequality follows from applying \cref{lem:fqi-approx} with a union bound. 
\end{proof}

In order to apply concave regression in estimating the optimal value function $Q_z^\star$ under the \MDP with transition kernel $\PP_z$, it is required that for any function $Q\in\scrC_{B,L}$, the Bellman optimality operator $\calT_z$ preserves the concavity, \ie, $\calT_z Q\in\scrC_{B,L}$.
As we show in \cref{lem:concavity-general,lem:concavity}, this holds for the Bellman optimality operator $\calT_z$ and $B$-bounded and $L$-Lipschitz concave function set $\scrC_{B,L}$ for $B = 1/(1-\gamma)$ and $L = L_r/(1-\gamma)$. 

\begin{lemma}\label{lem:concavity-general}
    Let $B \defeq 1/(1-\gamma)$ and $L \defeq L_r/(1-\gamma)$.
    For any mean-field term $z\in\calZ$ and the \MDP with stochastic concave transition kernel $\PP_z$, if any value function $Q\in\scrC_{B,L}$, then $\calT_z Q\in\scrC_{B,L}$.
    Consequently, the optimal value function $Q_z^\star\in\scrC_{B,L}$.
\end{lemma}

\begin{proof}
    For any bounded Lipschitz value function $Q\in\scrC_{B,L}$, we fix $b,b' \in \calB$, $w,w'\in\calW$, $a\in\Gamma_z(b,w)$, and $a'\in\Gamma_z(b',w')$. For any $\lambda \in[0,1]$, we define $\overline{b} \defeq (1-\lambda) b+\lambda b'$, $\overline{n} \defeq (1-\lambda) n+\lambda w'$, and $\overline{a} \defeq (1-\lambda) a+\lambda a'$. 
    Note that the state-action pair $(B,\overline w,\overline a)$ lies in the feasible set $\calG$ due to the concavity assumption.
    It then follows that under Bellman optimality operator $\calT_z$ that
    \begin{align*}
        (1-\lambda) (\calT_z Q)(b,w,a) + \lambda (\calT_z Q)(b',w',a') & = (1-\lambda) r_z(b,w,a) + \lambda r_z(b',w',a') \\
        & \qquad + \gamma (1-\lambda) \int_\calW \PP_z(\widetilde w \mid w) Q(a, \widetilde w, a_{1,\widetilde w}) \diff \widetilde w \\
        & \qquad + \gamma \lambda \int_\calW \PP_z(\widetilde w \mid w') Q(a', \widetilde w, a_{2,\widetilde w}) \diff \widetilde w,
    \end{align*}
    where we define $a_{1,\widetilde w} \defeq \argmax_{a^\ddagger\in\Gamma(a,\widetilde w)} Q(a, \widetilde w, a^\ddagger)$ to be the optimal action under state $(a,\widetilde w)$ and $a_{2,\widetilde w} \defeq \argmax_{a^\ddagger\in\Gamma(a',\widetilde w)} Q(a', \widetilde w, a^\ddagger)$ to be the optimal action under state $(a',\widetilde w)$ for any $\widetilde w\in\calW$.
    Given the optimal actions $a_{1,\widetilde w}$ and $a_{2,\widetilde w}$, we further denote $\widehat a \defeq (1-\lambda) a_{1,\widetilde w} + \lambda a_{2,\widetilde w}$.
    Notice that 
    \begin{align*}
        (1-\lambda) r_z(b,w,a) + \lambda r_z(b',w',a') \leq r_z(\overline{b},\overline w,\overline a)
    \end{align*}
    due to the concavity of the reward function $r_z$.
    It also follows from the stochastic concavity of the transition kernel $\PP_z$ that
    \begin{align*}
        & (1-\lambda) \int_\calW \PP_z(\widetilde w \mid w) Q(a, \widetilde w, a_{1,\widetilde w}) \diff \widetilde w + \lambda \int_\calW \PP_z(\widetilde w \mid w') Q(a', \widetilde w, a_{2,\widetilde w}) \diff \widetilde w \\
        & \qquad \leq \int_\calW \PP_z(\widetilde w \mid \overline w) Q(\overline a, \widetilde w, \widehat a) \diff \widetilde w.
    \end{align*}
    Hence, we have the concavity of $\calT_z Q$ following
    \begin{align*}
        (1-\lambda) (\calT_z Q)(b,w,a) + \lambda (\calT_z Q)(b',w',a') & \leq r_z(\overline{b},\overline w,\overline a)+\gamma \int_\calW \PP_z(\widetilde w \mid w) Q(\overline a, \widetilde w, \widehat a) \diff \widetilde w \\
        & \leq r_z(\overline{b},\overline w,\overline a)+\gamma \int_\calW \PP_z(\widetilde w \mid w) Q(\overline a, \widetilde w, a^\dagger) \diff \widetilde w \\
        & = (\calT_z Q)(\overline{b},\overline w,\overline a),
    \end{align*}
    where the second inequality is due to the the greedy action $Q(\overline a, \widetilde w, \widehat a) \leq \max_{a^\dagger} Q(\overline a, \widetilde w, a^\dagger)$, and the last equality is due to $a^\dagger \defeq \max_{a\in\Gamma(\overline a,\widetilde w)} Q(\overline a, \widetilde w, a^\dagger)$. The Lipschitzness of $\calT_z Q$ follows from $L = L_r/(1-\gamma)$.
        
    Recall that $Q_z^\star$ is the fixed point of $\calT_z$, and our desired conclusion follows from the closeness of $\scrC_{B,L}$ and the operation $\calT_z$ that preserves concavity and Lipschitzness of the value function.
\end{proof}

\begin{lemma}\label{lem:concavity}
    Let $B \defeq 1/(1-\gamma)$ and $L \defeq L_r/(1-\gamma)$.
    For any \MDP with transition kernel $\PP_z$ and any fixed income variable $w\in\calW$, if any value function $Q(\cdot,w,\cdot)\in\scrC_{B,L}$, then $\calT_z Q(\cdot,w,\cdot)\in\scrC_{B,L}$.
    Consequently, $Q_z^\star(\cdot,w,\cdot)\in\scrC_{B,L}$ for any $w \in \calW$.
\end{lemma}

\begin{proof}
    For any bounded Lipschitz value function $Q(\cdot,w,\cdot)\in\scrC_{B,L}$ with fixed $w\in\calW$, we pick any $b,b' \in \calB$, $a\in\Gamma_z(b,w)$, and $a'\in\Gamma_z(b',w)$. For any $\lambda \in[0,1]$, we define $\overline{b} \defeq (1-\lambda) b+\lambda b'$, and $\overline{a} \defeq (1-\lambda) a+\lambda a'$. 
    Note that the state-action pair $(\overline b,w,\overline a)$ lies in the feasible set $\calG$ due to the concavity assumption.
    It then follows that under Bellman optimality operator $\calT_z$
    \begin{align*}
        (1-\lambda) (\calT_z Q)(b,w,a) + \lambda (\calT_z Q)(b',w,a') & = (1-\lambda) r_z(b,w,a) + \lambda r_z(b',w,a') \\
        & \qquad + \gamma (1-\lambda) \int_\calW \PP_z(\widetilde w \mid w) Q(a, \widetilde w, a_{1,\widetilde w}) \diff \widetilde w \\
        & \qquad + \gamma \lambda \int_\calW \PP_z(\widetilde w \mid w) Q(a', \widetilde w, a_{2,\widetilde w}) \diff \widetilde w,
    \end{align*}
    where we define $a_{1,\widetilde w} \defeq \argmax_{a^\ddagger\in\Gamma(a,\widetilde w)} Q(a, \widetilde w, a^\ddagger)$ to be the optimal action under state $(a,\widetilde w)$ and $a_{2,\widetilde w} \defeq \argmax_{a^\ddagger\in\Gamma(a',\widetilde w)} Q(a', \widetilde w, a^\ddagger)$ to be the optimal action under state $(a',\widetilde w)$ for any $\widetilde w\in\calW$. Given the optimal actions $a_{1,\widetilde w}$ and $a_{2,\widetilde w}$, we further denote $\widehat a_{\widetilde w} \defeq (1-\lambda) a_{1,\widetilde w} + \lambda a_{2,\widetilde w}$.
    Notice that 
    \begin{align*}
        (1-\lambda) r_z(b,w,a) + \lambda r_z(b',w,a') \leq r_z(\overline{b},w,\overline a)
    \end{align*}
    due to the concavity of the reward function $r_z$.
    It also holds that
    \begin{align*}
        & (1-\lambda) \int_\calW \PP_z(\widetilde w \mid w) Q(a, \widetilde w, a_{1,\widetilde w}) \diff \widetilde w + \lambda \int_\calW \PP_z(\widetilde w \mid w) Q(a', \widetilde w, a_{2,\widetilde w}) \diff \widetilde w \\
        & \qquad = \int_\calW \PP_z(\widetilde w \mid w) [(1-\lambda) Q(a, \widetilde w, a_{1,\widetilde w}) \diff \widetilde w + \lambda Q(a', \widetilde w, a_{2,\widetilde w})] \diff \widetilde w \\
        & \qquad \leq \int_\calW \PP_z(\widetilde w \mid w) Q(\overline a, \widetilde w,\widehat a_{\widetilde w}) \diff \widetilde w,
    \end{align*}
    where the inequality is due to the concavity of the value function $Q$.
    Hence, we have the concavity of $\calT_z Q$ following
    \begin{align*}
        (\calT_z Q)(\overline{b},w,\overline a) & = r_z(\overline{b},w,\overline a)+\gamma \int_\calW \PP_z(\widetilde w \mid w) Q(\overline a, \widetilde w, a^\dagger) \diff \widetilde w \\
        & \geq r_z(\overline{b},w,\overline a)+\gamma \int_\calW \PP_z(\widetilde w \mid w) Q(\overline a, \widetilde w, \widehat a_{\widetilde w}) \diff \widetilde w,
    \end{align*}
    where the equality is due to the the greedy action $a^\dagger \defeq \max_{a\in\Gamma(\overline a,\widetilde w)} Q(\overline a, \widetilde w, a^\dagger)$, and the inequality is due to $Q(\overline a, \widetilde w, \widehat a) \leq  Q(\overline a, \widetilde w, a^\dagger)$. 
    
    Recall that $r_z(b,w,a)\in[0,1]$ for all $(b,w,a)\in\calG$. For any value function $Q$ such that $\|Q\|_\infty \leq \frac{1}{1-\gamma} = B$, we have $\|\calT_z Q\|_\infty \leq 1 + \frac{\gamma}{1-\gamma} = \frac{1}{1-\gamma} = B$ following the Bellman equation.
    The Lipschitzness of $\calT_z Q$ similarly follows from $L = L_r/(1-\gamma)$. Hence, we have $(\calT_z Q)(\cdot,w,\cdot)\in\scrC_{B,L}$ if any value function $Q(\cdot,w,\cdot)\in\scrC_{B,L}$.
    
    Notice that $Q_z^\star$ is the fixed point of $\calT_z$, and the operation $\calT_z$ preserving concavity and Lipschitzness of any value function $Q$ in $\scrC_{B,L}$ and the closeness of $\scrC_{B,L}$ lead to $Q_z^\star\in\scrC_{B,L}$.
\end{proof}

\subsection{Proof of FQI with Concave Regression}

\FQI is a popular offline \RL algorithm, together with its online value iteration counterpart, has been analyzed by a large body of literature under different settings~\citep{bertsekas2012dynamic,grunewalder2012modelling,fonteneau2013batch,sutton2018reinforcement,kumar2019stabilizing,zhu2020convex,min2021variance,fei2022cascaded,min2022learning,lu2023pessimism}. In particular, our method features \FQI with concave regression that integrates the economic insights to the algorithm.
 
We would like to introduce some notation for the proofs of \cref{algo:q-iteration} that follows. 
For any function $f$ and dataset $\calD_f = \{(X_m,Y_m)\}_{m=1}^M$ where data pairs $(X,Y)$ are \iid sampled from a distribution $\scrD$ such that $\expect_\scrD[Y\given X = x] = f(x)$, we define the empirical risk of any estimator $f'$ of the ground truth $f$ as 
\begin{align*}
    \calL_\calD(f',f) & \defeq \frac{1}{M} \sum_{m=1}^M (f'(X_m) - Y_m))^2,
\end{align*}
where the subscript $f$ on $\calD$ is omitted for simplicity.
Further, we define the true risk of the estimator $f'$ as 
\begin{align*}
    \calL(f',f) \defeq \expect[(f'(X) - Y))^2] = \expect_\calD[\calL_\calD(f',f)],
\end{align*}
where the expectation is taken with respect to the underlying distribution $\scrD$.
For any function class $\scrF$ and estimator $f'\in\scrF$, we denote 
\begin{align*}
    \calE_{\calD}(f',f) \defeq \calL_\calD(f',f) - \inf_{g\in\scrF} \calL_\calD(g,f)
\end{align*}
as the sub-optimality of $f'$ compared to the best empirical risk estimator of $f$ in $\scrF$.
We say $f'$ is an $\epsilon$-approximate \LSE if $\calE_{\calD}(f',f) \leq \epsilon$, and we denote the set of all $\epsilon$-approximate \LSEs $f'\in\scrF$ on data $\calD_f$ as $\mathrm{LSE}(\scrF,\calD_f;\epsilon)$.
Similarly, we also define 
\begin{align*}
    \calE(f',f) \defeq \calL(f',f) - \inf_{g\in\scrF} \calL(g,f)
\end{align*}
to be the sub-optimality of $f'$ in terms of true risk over distribution $\scrD$.
In particular, for any offline dataset $\calD_{\calT Q} \defeq \{(s_m,a_m,r_m,s'_m)\}_{m=1}^M$ \iid sampled from $(s,a)\sim\widetilde\mu$ under the transition kernel $s'\sim\PP(\cdot\given s,a)$, we aim to estimate $\calT Q$ with concave regression given any value function $Q\in\scrC_{B,L}$.
It follows that
\begin{align*}
    \calL_\calD(f',\calT Q) & = \frac{1}{M} \sum_{m=1}^M (f'(s_m,a_m) - r_m - \gamma V(s'_m))^2
\end{align*}
and for any function set $\scrF$,
\begin{align*}
    \calE_{\calD}(f',\calT Q) \defeq \calL_\calD(f',\calT Q) - \inf_{g\in\scrF} \calL_\calD(g,\calT Q),
\end{align*}
where $V(s'_m) \defeq \max_{a\in\Gamma(s'_m)} Q(s'_m,a)$ is the the action-value function under the greedy policy with respect to $Q$ evaluated at $s'_m$. Further, we have $\calL(f',\calT Q) \defeq \expect_\calD[\calL_\calD(f',\calT Q)]$.
We write $\calD_{\calT Q}$ as $\calD$ for simplicity when the context is clear.
Before getting into the proof for regression, we first bound the estimation error of \FQI using a similar argument as in \citet{chen2019information}.

\begin{lemma}\label{lem:fqi-error}
    Let $Q_z^\star$ be the optimal value function under the \MDP with transition kernel $\PP_z$ and $\widetilde Q_z^\ell$ be the estimation of $Q_z^\star$ at the $\ell$-th iteration of \cref{algo:q-iteration}. Then for any policy-induced measure $\mu$, it holds that the distance between $Q_z^\star$ and the estimator $\widetilde Q_z^\tau$ of \cref{algo:q-iteration} is bounded by
    \begin{align*}
        \|\widetilde Q_z^\tau - Q_z^\star\|_\mu \leq \sum_{\ell=1}^\tau \gamma^{\tau-\ell} \sqrt{D}\cdot \|\widetilde Q_z^\ell - \calT_z \widetilde Q_z^{\ell-1}\|_{\widetilde\mu} + \frac{\gamma^\tau}{1-\gamma}.
    \end{align*}
\end{lemma}

\begin{proof}
    Under the \MDP induced by mean-field term $z$ and policy-induced measure $\mu\in\Delta_\calG$, it holds for any $\ell\in[\tau]$ that
    \begin{align}\label{eqn:q-star-estimation}
        \|\widetilde Q_z^\ell - Q_z^\star\|_\mu & \leq \|\widetilde Q_z^\ell - \calT_z \widetilde Q_z^{\ell-1}\|_\mu + \|\calT_z \widetilde Q_z^{\ell-1} - Q_z^\star\|_\mu,
    \end{align}
    where $\calT_z Q$ for any value function $Q$ denotes the Bellman operator induced by the mean-field term $z$ and the greedy policy with respect to $Q$. 
    Notice that the first term of \eqref{eqn:q-star-estimation} represents the estimation error of $\widetilde Q_z^\ell$ to the target function $\calT_z \widetilde Q_z^{\ell-1}$.
    For any measure $\mu$, the second term of \eqref{eqn:q-star-estimation} can be further bounded through
    \begin{align*}
        \|\calT_z \widetilde Q_z^{\ell-1} - Q_z^\star\|_\mu^2 & = \|\calT_z \widetilde Q_z^{\ell-1} - \calT_z Q_z^\star\|_\mu^2 \\
        & = \expect_\mu[(\calT_z \widetilde Q_z^{\ell-1}(s,a) - \calT_z Q_z^\star(s,a))^2] \\
        & = \expect_\mu[(\gamma \PP[\widetilde V_z^{\ell-1}(s') - V_z^\star(s')])^2],
    \end{align*}
    where the first equality is due to the fact that the optimal value function $Q_z^\star$ is the fixed point of $\calT_z$, and the second equality follows from expanding the terms $\calT_z \widetilde Q_z^{\ell-1}(s,a) = r_z(s,a) + \gamma \int_\calS \PP_z(s'\given s,a) \widetilde V_z^{\ell-1}(s') \diff s$ and $\calT_z Q_z^\star(s,a) = r_z(s,a) + \gamma \int_\calS \PP_z(s'\given s,a) V_z^\star(s') \diff s$. Apply Jensen's inequality to the right-hand side of the second inequality, we have  
    \begin{align*}
        \|\calT_z \widetilde Q_z^{\ell-1} - Q_z^\star\|_\mu^2 & \leq \gamma^2 \expect_\mu \PP_z[(\widetilde V_z^{\ell-1}(s') - V_z^\star(s'))^2] \\
        & = \gamma^2 \expect_{\PP_z\mu}[(\widetilde V_z^{\ell-1}(s') - V_z^\star(s'))^2],
    \end{align*}
    where we define a shorthand $s'\sim\PP_z\mu$ that implies $s'\sim \PP_z(\cdot\given s,a)$ and $(s,a)\sim \mu$.
    In particular, we further expand according to the definitions of $\widetilde V_z^{\ell-1}(s')$ and $V_z^\star(s')$ to get
    \begin{align*}
        \expect_{\PP_z\mu}[(\widetilde V_z^{\ell-1}(s') - V_z^\star(s'))^2] & = \expect_{\PP_z\mu}\Big[\big(\max_{a\in\calA} \widetilde Q_z^{\ell-1}(s',a) - \max_{a\in\calA} Q_z^\star(s',a)\big)^2\Big] \\
        & \leq \expect_{\PP_z\mu}\Big[\big(\widetilde Q_z^{\ell-1}(s',\pi^\ddagger(\widetilde Q_z^{\ell-1},Q_z^\star)) - Q_z^\star(s',\pi^\ddagger(\widetilde Q_z^{\ell-1},Q_z^\star))\big)^2\Big] \\
        & = \|\widetilde Q_z^{\ell-1} - Q_z^\star\|_{\PP_z\mu\times\pi^\ddagger(\widetilde Q_z^{\ell-1},Q_z^\star)}^2,
    \end{align*}
    where the inequality follows from the inequality follows from the fact that $|\max_{x\in\calX} f(x) - \max_{y\in\calX} g(y)| \leq \max_{x\in\calX} |f(x) - g(x)|$ for any $x,y$ in the shared domain $\calX$. Especially, the policy $\pi^\ddagger(\widetilde Q_z^{\ell-1},Q_z^\star)$ is defined to take the action to maximize the difference between value functions $\widetilde Q_z^{\ell-1}$ and $Q_z^\star$, \ie, $\pi^\ddagger(\widetilde Q_z^{\ell-1},Q_z^\star)(s) \defeq \argmax_{a\in\calA} |\widetilde Q_z^{\ell-1}(s,a) - Q_z^\star(s,a)|$. Combine the upper bounds above together, for any $\ell\in[\tau]$, we have
    \begin{align}
        \|\widetilde Q_z^\ell - Q_z^\star\|_\mu & \leq \|\widetilde Q_z^\ell - \calT_z \widetilde Q_z^{\ell-1}\|_\mu + \gamma \|\widetilde Q_z^{\ell-1} - Q_z^\star\|_{\PP_z\mu\times\pi^\ddagger(\widetilde Q_z^{\ell-1},Q_z^\star)} \nonumber \\
        & \leq \sqrt{D}\|\widetilde Q_z^\ell - \calT_z \widetilde Q_z^{\ell-1}\|_{\widetilde\mu} + \gamma \|\widetilde Q_z^{\ell-1} - Q_z^\star\|_{\PP_z\mu\times\pi^\ddagger(\widetilde Q_z^{\ell-1},Q_z^\star)}, \label{eqn:q-estimation-recursion}
    \end{align}
    where the inequality follows from \cref{assum:concentrability}.
    Notice that $\PP_z\mu\times\pi^\ddagger(\widetilde Q_z^{\ell-1},Q_z^\star)$ is also a policy-induced measure, and we apply the above argument recursively on $\|\widetilde Q_z^{i-1} - Q_z^\star\|_{\PP_z\mu\times\pi_{\widetilde Q_z^{i-1},Q_z^\star}}$ for all $i\in[\ell]$. Take $\ell = \tau$ and expand the inequality \eqref{eqn:q-estimation-recursion} for $\tau$ times, we obtain
    \begin{align*}
        \|\widetilde Q_z^\tau - Q_z^\star\|_\mu & \leq \sum_{\ell=1}^\tau \gamma^{\tau-\ell} \sqrt{D}\cdot \|\widetilde Q_z^\ell - \calT_z \widetilde Q_z^{\ell-1}\|_{\widetilde\mu} + \frac{\gamma^\tau}{1-\gamma},
    \end{align*}
    where $\|\widetilde Q_z^\ell - \calT_z \widetilde Q_z^{\ell-1}\|_{\widetilde\mu}$ is the estimation error for target function $\calT_z \widetilde Q_z^{\ell-1}$.
\end{proof}

\begin{lemma}\label{lem:approx-error}
    Let $B \defeq 1/(1-\gamma)$ and $L \defeq L_r/(1-\gamma)$, and for any mean-field term $z$, $Q\in\scrC_{B,L}$, and data $\calD_{\calT_z Q} \defeq \{(s_m,a_m,r_m,s'_m)\}_{m=1}^M$, we define $\widehat f\in\mathrm{LSE}(\scrC_{B,L},\calD_{\calT_z Q};\epsilon)$. Then for any measure $\widetilde\mu\in\Delta_\calG$ and any training set size $M \defeq |\calD_{\calT_z Q}|$ such that $M^\frac{d}{d+2} \gtrsim \log M$, $M^\frac{d}{d+2}\log M \gtrsim B^\frac{3d}{d+2} (d+1)^\frac{-2}{d+2} (dL)^\frac{-d}{d+2} \log\frac{1}{\delta}$, and  
    \begin{align*}
        (\log M)^\frac{d}{2} \gtrsim 1 + \log(R_d^*B^\frac{-6}{d+2} (d+1)^\frac{-2}{d+2} (dL)^\frac{-d}{d+2}),
    \end{align*}
    with probability at least $1-\delta$
    \begin{align*}
        \|\widehat f - \calT_z Q\|_{\widetilde\mu}^2 \lesssim B^\frac{d+8}{d+2} (d+1)^\frac{2}{d+2} (dL)^\frac{d}{d+2} M^{-\frac{2}{d+2}}\log M + \epsilon,
    \end{align*}
    where $R_d^* \leq \max\{8dL, 2B+4dL\}$.
\end{lemma}

\begin{proof}
    Before getting into the proof of the lemma, we first provide a few useful definitions. We denote $g_z^\dagger(Q) \defeq \argmin_{f\in\scrC_{B,L}} \calL(f,\calT_z Q)$ to be the \LSE of $\calT_z Q$ that minimizes the true risk under $\widetilde\mu$ for any mean-field term $z$ and value function $Q\in\scrC_{B,L}$. 
    We also define $g_z(Q;\calD) \defeq \argmin_{f\in\scrC_{B,L}} \calL_\calD(f,\calT_z Q)$ to be the \LSE of $\calT_z Q$ that minimizes the empirical risk on $\calD$.
    Note that $\calT_z Q\in\scrC_{B,L}$ following \cref{lem:concavity-general}, and we have $g_z^\dagger(Q) = \calT_z Q$. Further recall that 
    \begin{align*}
        \calL_\calD(\widehat f,\calT Q) & = \frac{1}{M} \sum_{m=1}^M (f'(s_m,a_m) - r_m - \gamma \max_{a\in\Gamma(s'_m)} Q(s'_m,a))^2
    \end{align*}
    and $\calL(\widehat f,\calT Q) = \expect_\calD[\calL_\calD(\widehat f,\calT Q)]$. 
    For any $\widehat f\in\mathrm{LSE}(\scrC_{B,L},\calD;\epsilon)$, it follows by definition that $\calE_{\calD}(\widehat f,\calT Q) = \calL_\calD(\widehat f,\calT Q) - \calL_\calD(g_z(Q;\calD),\calT Q) \leq \epsilon$.
    
    It is noteworthy that $\|\widehat f - \calT_z Q\|_{\widetilde\mu}^2 - \|g_z^\dagger(Q) - \calT_z Q\|_{\widetilde\mu}^2 = \calL(\widehat f,\calT_z Q) - \calL(g_z^\dagger(Q),\calT_z Q)$, which is followed from the fact that for any $f\in\scrC_{B,L}$, 
    \begin{align*}
        \calL(f,\calT_z Q) & = \expect_{\widetilde\mu}[(f(s,a) - r_z(s,a) - \gamma (\frakJ Q)(s'))^2] \\
        & = \expect_{\widetilde\mu}[(f(s,a) - (\calT_z Q)(s,a))^2] + \Var_{\widetilde\mu}[\expect[r(s,a) + \gamma (\frakJ Q)(s')\given s,a]],
    \end{align*}
    where $(s,a)\sim\widetilde\mu$ and $s'\sim\PP_z(\cdot\given s,a)$; the variance terms cancel with each other and yield the equality.
    It follows that
    \begin{align*}
        \|\widehat f - \calT_z Q\|_{\widetilde\mu}^2 & =  \calL(\widehat f,\calT_z Q) - \calL(\calT_z Q,\calT_z Q)
    \end{align*}
    where the equality is due to $\|g_z^\dagger(Q) - \calT_z Q\|_{\widetilde\mu}^2 = 0$ and $g_z^\dagger(Q) = \calT_z Q$.
    Following \cref{thm:empirical-bound}, it holds that for any $\alpha,\beta > 0$ and $0 < \varrho \leq \frac{1}{2}$,
    \begin{align*}
        & \prob\Big\{\exists f \in \scrC_{B,L}: \calL(f,\calT_z Q) - \calL(\calT_z Q,\calT_z Q) - (\calL_\calD(f,\calT_z Q) - \calL_\calD(\calT_z Q,\calT_z Q)) \\
        & \qquad\qquad\qquad \geq \varrho \cdot(\alpha+\beta+\calL(f,\calT_z Q) - \calL(\calT_z Q,\calT_z Q))\Big\}, \\
        & \qquad \leq 14 \sup_\calD \calN_1\left(\frac{\beta \varrho}{20B}, \scrC_{B,L}, \calD\right) \exp \left(-\frac{\varrho^2(1-\varrho) \alpha M}{214(1+\varrho)B^4}\right),
    \end{align*}
    where the inequality follows from $|f| \leq \frac{1}{1-\gamma} = B$ for all $f\in\scrC_{B,L}$, and $\calN_1\left(\eps, \scrF, \calD\right)$ denotes the $\eps$-covering number of $\scrF$ on $\calD$ with respect to $\ell_1$ metric. 
    The supremum is taken with respect to all random dataset $\calD$.
    This provides a high probability upper bound on $\|\widehat f - \calT_z Q\|_{\widetilde\mu}^2$ as 
    \begin{align*}
        & \prob\Big\{(1-\varrho)[\calL(\widehat f,\calT_z Q) - \calL(\calT_z Q,\calT_z Q)] \geq \varrho \cdot(\alpha+\beta) + [\calL_\calD(\widehat f,\calT_z Q) - \calL_\calD(\calT_z Q,\calT_z Q)]\Big\} \\
        & \qquad \leq \prob\Big\{\exists f \in \scrC_{B,L}: \calL(f,\calT_z Q) - \calL(g^\dagger,\calT_z Q) - (\calL_\calD(f,\calT_z Q) - \calL_\calD(g^\dagger,\calT_z Q)) \\
        & \qquad\qquad\qquad \geq \varrho \cdot(\alpha+\beta+\calL(f,\calT_z Q) - \calL(g^\dagger,\calT_z Q))\Big\}.
    \end{align*}
    Further, for the set of all $B$-bounded, $L$-Lipschitz concave functions, its covering entropy is bounded by 
    \begin{align*}
        \sup_\calD\log \calN_1\left(\frac{\beta \varrho}{20B}, \scrC_{B,L}, \calD\right) & \leq \log \calN_\infty\left(\frac{\beta \varrho}{20B}, \scrC_{B,L}\right) \\
        & \leq 2(d+1)\left(\frac{1600 BdL}{\beta \varrho}\right)^{d / 2} \ln \left(\frac{200B R_d^*}{\beta \varrho}\right),
    \end{align*}
    where the first inequality follows from upper bounding $\ell_1$ covering number with $\ell_\infty$ covering number, and the second inequality follows from \cref{lem:covering-convex}. More specifically, for Lipschitz constant $L$ and uniform $\ell_\infty$ bound $B$, the covering entropy of $\scrC_{B,L}$ is bounded through
    \begin{align*}
        \log\calN_\infty(\eps, \scrC_{B,L}) \leq 2(d+1)\left(\frac{80 L_d}{\eps}\right)^{d / 2} \ln \left(\frac{10 R_d^*}{\eps}\right)
    \end{align*}
    for any $\eps\in(0,80L_d]$ where $L_d \defeq dL\diam(\calG) \leq dL$ and $R_d^* \leq \max\{8dL, 2B+4dL\}$.
    Thus, for any $\beta \leq \frac{1600}{\varrho}BdL$, we have
    \begin{align*}
        & \prob\Big\{(1-\varrho)[\calL(\widehat f,\calT_z Q) - \calL(\calT_z Q,\calT_z Q)] \geq \varrho \cdot(\alpha+\beta) + [\calL_\calD(\widehat f,\calT_z Q) - \calL_\calD(\calT_z Q,\calT_z Q)]\Big\} \\
        & \qquad \leq 14 \exp\left(2(d+1)\left(\frac{1600 BdL}{\beta \varrho}\right)^{d / 2} \ln \left(\frac{200B R_d^*}{\beta \varrho}\right) - \frac{\varrho^2(1-\varrho) \alpha M}{214(1+\varrho)B^4}\right).
    \end{align*}
    Take $\varrho = \frac{1}{2}$ and for the estimator $\widehat f\in\scrC_{B,L}$ that with probability at least 
    \begin{align*}
        1 - 14 \exp\left(2(d+1)\left(\frac{3200 BdL}{\beta}\right)^{d / 2} \ln \left(\frac{400B R_d^*}{\beta}\right) - \frac{\alpha M}{2568B^4}\right),
    \end{align*}
    it holds that the estimation error is bounded by
    \begin{align*}
        \|\widehat f - \calT_z Q\|_{\widetilde\mu}^2 & = \calL(\widehat f,\calT_z Q) - \calL(\calT_z Q,\calT_z Q) \\
        & \leq \alpha+\beta + 2 (\calL_\calD(\widehat f,\calT_z Q) - \calL_\calD(\calT_z Q,\calT_z Q)) \\
        & = \alpha+\beta + 2 (\calE_{\calD}(\widehat f,\calT Q) + \calL_\calD(g_z(Q;\calD),\calT Q) - \calL_\calD(\calT_z Q,\calT_z Q)) \\
        & \leq \alpha + \beta + 2\epsilon
    \end{align*}
    where the second equality is due to $\calE_{\calD}(\widehat f,\calT Q) = \calL_\calD(\widehat f,\calT Q) - \calL_\calD(g_z(Q;\calD),\calT Q)$, and the last inequality is due to the fact that $g_z(Q;\calD)$ is defined to be the \LSE that minimizes the empirical risk, \ie, $\calL_\calD(g_z(Q;\calD),\calT_z Q) \leq \calL_\calD(\calT_z Q,\calT_z Q)$.
    Consequently, solve for proper $\alpha$ and $\beta$ to get our conclusion: if $M^\frac{d}{d+2} \gtrsim \log M$, $M^\frac{d}{d+2}\log M \gtrsim B^\frac{3d}{d+2} (d+1)^\frac{-2}{d+2} (dL)^\frac{-d}{d+2} \log\frac{1}{\delta}$, and  
    \begin{align*}
        (\log M)^\frac{d}{2} \gtrsim 1 + \log(R_d^*B^\frac{-6}{d+2} (d+1)^\frac{-2}{d+2} (dL)^\frac{-d}{d+2}),
    \end{align*}
    then for any $d\geq 2$, with probability at least $1-\delta$ 
    \begin{align*}
        \|\widehat f - \calT_z Q\|_{\widetilde\mu}^2 & \lesssim B^\frac{d+8}{d+2} (d+1)^\frac{2}{d+2} (dL)^\frac{d}{d+2} M^{-\frac{2}{d+2}}\log M + \epsilon.
    \end{align*}
\end{proof}

\begin{lemma}\label{lem:fqi-approx}
    For the \MDP with transition kernel $\PP_z$, if the regression of \cref{algo:q-iteration} finds $\epsilon$-approximate \LSE in $\scrC_{B,L}$ at all $\tau$ iterations, then the total approximation error of $\widetilde Q_z^\tau$ with respect to the optimal value function $Q_z^\star$ under any measure $\widetilde\mu$ is upper bounded by
    \begin{align*}
        \|\widetilde Q_z^\tau - Q_z^\star\|_{\widetilde\mu} & \lesssim \frac{\sqrt{D}}{1-\gamma} B^\frac{d+8}{d+2} (d+1)^\frac{2}{d+2} (dL)^\frac{d}{d+2} M^{-\frac{2}{d+2}}\log M + \frac{\sqrt{D}\epsilon}{1-\gamma}
    \end{align*}
    with probability at least $1-\delta$ for the number of iterations 
    \begin{align*}
        \tau \asymp \log\left( D^{-1/2} B^{-\frac{d+8}{d+2}}(d+1)^{-\frac{2}{d+2}}(dL)^{-\frac{d}{d+2}}M^\frac{2}{d+2}\frac{1}{\log M} \right) / \log\frac{1}{\gamma}
    \end{align*}
    for any sample size $M^\frac{d}{d+2} \gtrsim \log M$, $M^\frac{d}{d+2}\log M \gtrsim B^\frac{3d}{d+2} (d+1)^\frac{-2}{d+2} (dL)^\frac{-d}{d+2} \log\frac{\tau}{\delta}$, and  
    \begin{align*}
        (\log M)^\frac{d}{2} \gtrsim 1 + \log(R_d^*B^\frac{-6}{d+2} (d+1)^\frac{-2}{d+2} (dL)^\frac{-d}{d+2}),
    \end{align*}
    where $R_d^* \leq \max\{8dL, 2B+4dL\}$.
\end{lemma}

\begin{proof}
    Combining \cref{lem:fqi-error} and \cref{lem:approx-error}, we have with probability at least $1-\delta$ that the total approximation error of the output $\widetilde Q_z^\tau$ to the optimal value function $Q_z^\star$ under any measure $\widetilde\mu$ is upper bounded by
    \begin{align*}
        \|\widetilde Q_z^\tau - Q_z^\star\|_{\widetilde\mu} & \leq \sum_{\ell=1}^\tau \gamma^{\tau-\ell} \sqrt{D}\cdot \|\widetilde Q_z^\ell - \calT \widetilde Q_z^{\ell-1}\|_{\widetilde\mu} + \frac{\gamma^\tau}{1-\gamma} \\
        & \lesssim \sum_{\ell=1}^\tau \gamma^{\tau-\ell} \sqrt{D}\cdot \left(B^\frac{d+8}{d+2} (d+1)^\frac{2}{d+2} (dL)^\frac{d}{d+2} M^{-\frac{2}{d+2}}\log M + \epsilon\right) + \frac{\gamma^\tau}{1-\gamma} \\
        & \lesssim \frac{1}{1-\gamma} \sqrt{D}\cdot \left(B^\frac{d+8}{d+2} (d+1)^\frac{2}{d+2} (dL)^\frac{d}{d+2} M^{-\frac{2}{d+2}}\log M + \epsilon\right) + \frac{\gamma^\tau}{1-\gamma},
    \end{align*}
    for $M^\frac{d}{d+2} \gtrsim \log M$, $M^\frac{d}{d+2}\log M \gtrsim B^\frac{3d}{d+2} (d+1)^\frac{-2}{d+2} (dL)^\frac{-d}{d+2} \log\frac{\tau}{\delta}$, and  
    \begin{align*}
        (\log M)^\frac{d}{2} \gtrsim 1 + \log(R_d^*B^\frac{-6}{d+2} (d+1)^\frac{-2}{d+2} (dL)^\frac{-d}{d+2}).
    \end{align*} 
    The first inequality follows from \cref{lem:fqi-error}, the second inequality is due to applying \cref{lem:approx-error} such that the upper bound on $\|\widetilde Q_z^\ell - \calT \widetilde Q_z^{\ell-1}\|_{\widetilde\mu}$ holds with probability at least $1-\delta/\tau$ for all $\ell\in[\tau]$, and the last inequality follows from a sum over geometric series.
    If the number of iteration $\tau$ is large enough, such that
    \begin{align*}
        \tau & \gtrsim \log\left( D^{-1/2} B^{-\frac{d+8}{d+2}}(d+1)^{-\frac{2}{d+2}}(dL)^{-\frac{d}{d+2}}M^\frac{2}{d+2}\frac{1}{\log M} \right) / \log\frac{1}{\gamma},
    \end{align*}
    it follows that
    \begin{align*}
        \|\widetilde Q_z^\tau - Q_z^\star\|_{\widetilde\mu} & \lesssim \frac{\sqrt{D}}{1-\gamma} B^\frac{d+8}{d+2} (d+1)^\frac{2}{d+2} (dL)^\frac{d}{d+2} M^{-\frac{2}{d+2}}\log M + \frac{\sqrt{D}\epsilon}{1-\gamma}.
    \end{align*}
\end{proof}

\subsection{Proof of FQI With Max-Affine Functions}

In this subsection, we present the analysis of our \FQI algorithm when the underlying feasible functions are max-affine.
Max-affine functions generalize the commonly studied linear functions in \RL~\citep{jin2020provably,ayoub2020model,cai2020provably,du2021bilinear,min2022learn,feicascaded}, and more importantly, all max-affine functions constitute a subset of all convex functions.
Let $\scrA_{B,L}^{K}$ denote the set of all bounded $L$-Lipschitz $K$-max-affine functions defined as follows
\begin{align*}
    \scrA_{B,L}^{K} \defeq \{h:\calG\to\RR \given h(x) = \max_{k\in[K]} \alpha_k\transpose x + c_k, \|\alpha_k\|_\infty \leq L, h(x)\in [-dL\diam(\calG),B]\}.
\end{align*}

\begin{lemma}\label{lem:approx-error-affine}
    Let $B \defeq 1/(1-\gamma)$ and $L \defeq L_r/(1-\gamma)$, and for any mean-field term $z$, $Q\in\scrC_{B,L}$, and dataset $\calD_{\calT_z Q} \defeq \{(s_m,a_m,r_m,s'_m)\}_{m=1}^M$ with $(s_m,a_m)\in\Rd$ \iid sampled from any measure $\widetilde\mu\in\Delta_\calG$ for all $m\in[M]$, we define the $\epsilon$-approximate estimator $\widehat f\in\mathrm{LSE}(\scrA_{B,L}^{K,+},\calD_{\calT_z Q};\epsilon)$ of $\calT_z Q$ for $K\in\ZZ_+$. Then for any training set size $M \defeq |\calD_{\calT_z Q}|$ such that
    \begin{align*}
        M \gtrsim \frac{1}{e^{2(d+1)M^{d/(d+4)}}\delta}, \qquad M^\frac{2d+4}{d+4} \log M \gtrsim \frac{B^3(B+4dL)}{d+1}, 
    \end{align*}
    then with probability at least $1-\delta$
    \begin{align*}
        \|\widehat f - \calT_z Q\|_{\widetilde\mu}^2 & \lesssim (d+1)B^4 M^{-4/(d+4)} \log M + d^2L^2 M^{-4/(d+4)} \\
        & \qquad + BdL M^{-4/(d+4)} \sqrt{\log\frac{2}{\delta}} + \frac{B^2\log\frac{2}{\delta}}{M} + \epsilon.
    \end{align*}
\end{lemma}

\begin{proof}
    Let us denote $g_z^\dagger(Q) \defeq \argmin_{f\in\scrA_{B,L}^{K,+}} \calL(f,\calT_z Q)$ to be the estimator of $\calT_z Q$ in $\scrA_{B,L}^{K,+}$ that minimizes the true risk under $\widetilde\mu$ for any mean-field term $z$. 
    We also define $g_z(Q;\calD) \defeq \argmin_{f\in\scrA_{B,L}^{K,+}} \calL_\calD(f,\calT_z Q)$ to be the \LSE of $\calT_z Q$ in $\scrA_{B,L}^{K,+}$ that minimizes the empirical risk on $\calD_{\calT_z Q}$.
    Note that for any $Q\in\scrC_{B,L}$, $\calT_z Q\in\scrC_{B,L}$ following \cref{lem:concavity-general}, but it does not have to be in $\scrA_{B,L}^{K,+}$. 
    Hence, the estimator $g_z^\dagger(Q) = \calT_z Q$ may not hold. 
    For any $\widehat f\in\mathrm{LSE}(\scrA_{B,L}^{K,+},\calD_{\calT_z Q};\epsilon)$, it follows by definition that $\calE_{\calD}(\widehat f,\calT_z Q) = \calL_\calD(\widehat f,\calT_z Q) - \calL_\calD(g_z(Q;\calD),\calT_z Q) \leq \epsilon$.
    
    Recall that $\|\widehat f - \calT_z Q\|_{\widetilde\mu}^2 = \calL(\widehat f,\calT_z Q) - \calL(\calT_z Q,\calT_z Q)$.
    Following \cref{thm:empirical-bound}, it holds that for any $\alpha,\beta > 0$ and $0 < \varrho \leq \frac{1}{2}$,
    \begin{align*}
        & \prob\Big\{\exists f \in \scrA_{B,L}^{K,+}: \calL(f,\calT_z Q) - \calL(\calT_z Q,\calT_z Q) - (\calL_\calD(f,\calT_z Q) - \calL_\calD(\calT_z Q,\calT_z Q)) \\
        & \qquad\qquad\qquad \geq \varrho \cdot(\alpha+\beta+\calL(f,\calT_z Q) - \calL(\calT_z Q,\calT_z Q))\Big\}, \\
        & \qquad \leq 14 \sup_\calD \calN_1\left(\frac{\beta \varrho}{20B}, \scrA_{B,L}^{K,+}, \calD\right) \exp \left(-\frac{\varrho^2(1-\varrho) \alpha M}{214(1+\varrho)B^4}\right),
    \end{align*}
    where the inequality follows from $|f| \leq B$ for all $f\in\scrA_{B,L}^{K,+}$, and $\calN_1\left(\eps, \scrF, \calD\right)$ denotes the $\eps$-covering number of $\scrF$ on $\calD$ with respect to $\ell_1$ metric. 
    The supremum is taken with respect to all possible data $\calD$.
    This provides a high probability upper bound on $\|\widehat f - \calT_z Q\|_{\widetilde\mu}^2$ as 
    \begin{align*}
        & \prob\Big\{(1-\varrho)[\calL(\widehat f,\calT_z Q) - \calL(\calT_z Q,\calT_z Q)] \geq \varrho \cdot(\alpha+\beta) + [\calL_\calD(\widehat f,\calT_z Q) - \calL_\calD(\calT_z Q,\calT_z Q)]\Big\} \\
        & \qquad \leq \prob\Big\{\exists f \in \scrA_{B,L}^{K,+}: \calL(f,\calT_z Q) - \calL(\calT_z Q,\calT_z Q) - (\calL_\calD(f,\calT_z Q) - \calL_\calD(\calT_z Q,\calT_z Q)) \\
        & \qquad\qquad\qquad \geq \varrho \cdot(\alpha+\beta+\calL(f,\calT_z Q) - \calL(\calT_z Q,\calT_z Q))\Big\}.
    \end{align*}
    Further, for the set of all $B$-bounded, $L$-Lipschitz $K$-piece max-affine functions, its covering entropy is bounded by 
    \begin{align*}
        \sup_\calD\log \calN_1\left(\frac{\beta \varrho}{20B}, \scrA_{B,L}^{K,+}, \calD\right) & \leq \log \calN_\infty\left(\frac{\beta \varrho}{20B}, \scrA_{B,L}^K\right) \\
        & \leq (d+1) K \log\frac{(20B+80dL)B}{\beta \varrho},
    \end{align*}
    where the first inequality follows from upper bounding $\ell_1$ covering number with $\ell_\infty$ covering number, and the second inequality follows from \cref{lem:covering-affine} with $L_d \leq dL$. More specifically, for Lipschitz constant $L$ and uniform $\ell_\infty$ bound $B$, the covering entropy of $\scrA_{B,L}^K$ is bounded through
    \begin{align*}
        \log\calN_\infty(\eps,\scrA_{B,L}^K) \leq (d+1) K \log\frac{B+4L_d}{\eps}
    \end{align*}
    for any $\eps \leq B+4L_d$.
    For our regression problem, $d$ is the dimension of $\calG$. 
    Thus, for any $\beta \leq \frac{20BB}{\varrho}(4dL + B)$, we have
    \begin{align*}
        & \prob\Big\{(1-\varrho)[\calL(\widehat f,\calT_z Q) - \calL(\calT_z Q,\calT_z Q)] \geq \varrho \cdot(\alpha+\beta) + [\calL_\calD(\widehat f,\calT_z Q) - \calL_\calD(\calT_z Q,\calT_z Q)]\Big\} \\
        & \qquad \leq 14 \exp\left((d+1) K \log\frac{(20B+80dL)B}{\beta \varrho} - \frac{\varrho^2(1-\varrho) \alpha M}{214(1+\varrho)B^4}\right).
    \end{align*}
    Take $\varrho = \frac{1}{2}$ and for the estimator $\widehat f\in\scrA_{B,L}^K$ that with probability at least 
    \begin{align*}
        1 - 14 \exp\left((d+1) K \log\frac{(40B+160dL)B}{\beta} - \frac{\alpha M}{2568B^4}\right),
    \end{align*}
    it holds that the estimation error is bounded by
    \begin{align*}
        \|\widehat f - \calT_z Q\|_{\widetilde\mu}^2 & = \calL(\widehat f,\calT_z Q) - \calL(\calT_z Q,\calT_z Q) \\
        & \leq \alpha+\beta + 2 (\calL_\calD(\widehat f,\calT_z Q) - \calL_\calD(\calT_z Q,\calT_z Q)) \\
        & = \alpha+\beta + 2 (\calE_{\calD}(\widehat f,\calT_z Q) + \calL_\calD(g_z(Q;\calD),\calT_z Q) - \calL_\calD(\calT_z Q,\calT_z Q)) \\
        & \leq \alpha + \beta + 2\epsilon + 2(\calL_\calD(g_z(Q;\calD),\calT_z Q) - \calL_\calD(\calT_z Q,\calT_z Q)) \\
        & \leq \alpha + \beta + 2\epsilon + 2(\calL_\calD(g_z^\dagger(Q),\calT_z Q) - \calL_\calD(\calT_z Q,\calT_z Q)),
    \end{align*}
    where the second equality is due to $\calE_{\calD}(\widehat f,\calT_z Q) = \calL_\calD(\widehat f,\calT_z Q) - \calL_\calD(g_z(Q;\calD),\calT_z Q)$ and the last inequality is due to the fact that $g_z(Q;\calD)$ is defined to be the \LSE that minimizes the empirical risk, \ie, $\calL_\calD(g_z(Q;\calD),\calT_z Q) \leq \calL_\calD(g_z^\dagger(Q),\calT_z Q)$.
    Consequently, solve for proper $\alpha$ and $\beta$ to get that if  
    \begin{align*}
        M \gtrsim \frac{1}{e^{2(d+1)K}\delta}, \qquad M\log M \gtrsim \frac{(B+4dL)B^3}{(d+1)K},
    \end{align*}
    then for any $d\geq 2$, with probability at least $1-\delta/2$ 
    \begin{align*}
        \|\widehat f - \calT_z Q\|_{\widetilde\mu}^2 & \lesssim \frac{(d+1)K}{B^4} \frac{\log M}{M} + 2\epsilon + 2(\calL_\calD(g_z^\dagger(Q),\calT_z Q) - \calL_\calD(\calT_z Q,\calT_z Q)).
    \end{align*}
    
    Following \cref{lem:approx-affine}, under any measure $\widetilde\mu$ it holds that
    \begin{align*}
        \|g_z^\dagger(Q) - \calT_z Q\|_{\widetilde\mu}^2 \leq \|\overline g_z^\ddagger(Q) - \calT_z Q\|_{\widetilde\mu}^2 \leq 72^2 L_d^2 (K-1)^{-4/d},
    \end{align*}
    where $\overline g_z^\ddagger(Q) \defeq \sigma(g_z^\ddagger(Q))$ is the truncated estimator, where $\sigma$ denotes ReLU function and $g_z^\ddagger(Q) \defeq \argmin_{f\in\scrA_{B,L}^{K-1}} \calL(f,\calT_z Q)$.
    Note that $\|g_z^\dagger(Q) - \calT_z Q\|_{\widetilde\mu}^2 = \calL(g_z^\dagger(Q),\calT_z Q) - \calL(\calT_z Q,\calT_z Q)$, and following Bernstein's inequality, we have
    \begin{align*}
        & |\calL(g_z^\dagger(Q),\calT_z Q) - \calL(\calT_z Q,\calT_z Q) - \calL_\calD(g_z^\dagger(Q),\calT_z Q) + \calL_\calD(\calT_z Q,\calT_z Q)| \\
        & \qquad = \Big|\calL(g_z^\dagger(Q),\calT_z Q) - \calL(\calT_z Q,\calT_z Q) - \frac{1}{n}\sum_{(s,a,r,s')\in\calD} [(g_z^\dagger(Q)(s,a) - (\calT_z Q)(s,a)) \\
        & \qquad\qquad \cdot(g_z^\dagger(Q)(s,a) + (\calT_z Q)(s,a) - 2r - 2\gamma (\frakJ Q)(s'))] \Big| \\
        & \qquad \leq \frac{2B^2\log\frac{1}{\delta}}{3M} + 2\sqrt{\frac{\log\frac{1}{\delta}}{M}} \cdot \Var([(g_z^\dagger(Q)(s,a) - (\calT_z Q)(s,a)) \\
        & \qquad\qquad \cdot(g_z^\dagger(Q)(s,a) + (\calT_z Q)(s,a) - 2r - 2\gamma (\frakJ Q)(s'))])^{1/2}.
    \end{align*}
    Since we have $|g_z^\dagger(Q)(s,a) + (\calT_z Q)(s,a) - 2r - 2\gamma (\frakJ Q)(s')| \leq 2 B$ and $\Var(X) = \expect X^2 - (\expect X)^2$ for any random variable $X$, we further simplify the upper bound into
    \begin{align*}
        & |\calL(g_z^\dagger(Q),\calT_z Q) - \calL(\calT_z Q,\calT_z Q) - \calL_\calD(g_z^\dagger(Q),\calT_z Q) + \calL_\calD(\calT_z Q,\calT_z Q)| \\
        & \qquad \leq \frac{2B^2\log\frac{1}{\delta}}{3M} + 4B \|g_z^\dagger(Q) - \calT_z Q\|_{\widetilde\mu} \sqrt{\frac{\log\frac{1}{\delta}}{M}}.
    \end{align*}
    It follows that with probability at least $1-\delta/2$ 
    \begin{align*}
        \calL_\calD(g_z^\dagger(Q),\calT_z Q) - \calL_\calD(\calT_z Q,\calT_z Q) & \leq \|g_z^\dagger(Q) - \calT_z Q\|_{\widetilde\mu}^2 + 4B \|g_z^\dagger(Q) - \calT_z Q\|_{\widetilde\mu} \sqrt{\frac{\log\frac{2}{\delta}}{M}} \\
        & \qquad + \frac{2B^2\log\frac{2}{\delta}}{3M} \\
        & \leq 72^2 L_d^2 (K-1)^{-4/d} + 288BL_d K^{-2/d} \sqrt{\frac{\log\frac{2}{\delta}}{M}} + \frac{2B^2\log\frac{2}{\delta}}{3M}.
    \end{align*}
    
    Combine the upper bounds above together, we have with probability at least $1-\delta$
    \begin{align*}
        \|\widehat f - \calT_z Q\|_{\widetilde\mu}^2 & \lesssim (d+1)KB^4 \frac{\log M}{M} + d^2L^2 (K-1)^{-4/d} + BdL K^{-2/d} \sqrt{\frac{\log\frac{2}{\delta}}{M}} + \frac{B^2\log\frac{2}{\delta}}{M} + \epsilon.
    \end{align*}
    Take $K = \lceil M^{d/(d+4)}\rceil$, then it holds that
    \begin{align*}
        \|\widehat f - \calT_z Q\|_{\widetilde\mu}^2 & \lesssim (d+1)B^4 M^{-4/(d+4)} \log M + d^2L^2 M^{-4/(d+4)} \\
        & \qquad + BdL M^{-4/(d+4)} \sqrt{\log\frac{2}{\delta}} + \frac{B^2\log\frac{2}{\delta}}{M} + \epsilon.
    \end{align*}
\end{proof}

\begin{lemma}\label{lem:fqi-approx-affine}
    Let $B \defeq 1/(1-\gamma)$ and $L \defeq L_r/(1-\gamma)$.
    For the \MDP with transition kernel $\PP_z$, if the regression of \cref{algo:q-iteration} finds $\epsilon$-approximate \LSE in $\scrA_{B,L}^K$ with $K = \lceil M^{d/(d+4)}\rceil$ at all $\tau$ iterations, then the approximation error of $\widetilde Q_z^\tau$ with respect to the optimal value function $Q_z^\star$ under any measure $\widetilde\mu$ is upper bounded by
    \begin{align*}
        \|\widetilde Q_z^\tau - Q_z^\star\|_{\widetilde\mu} & \lesssim \frac{\sqrt{D}}{1-\gamma} \left( (d+1)B^4 + (B+dL)dL \right) M^{-\frac{4}{d+4}}\log M + \frac{\sqrt{D}\epsilon}{1-\gamma}
    \end{align*}
    with probability at least $1-\delta$ for the number of iterations 
    \begin{align*}
        \tau \asymp \log \frac{M^{4/(d+4)}}{d^2L^2\sqrt{D}} / \log\frac{1}{\gamma}
    \end{align*}
    for any sample size $M = |\calD|$ such that
    \begin{gather*}
        M \gtrsim \frac{1}{\delta\log\frac{1}{\gamma}} \log \frac{M^{4/(d+4)}}{d^2L^2\sqrt{D}}, \qquad M^\frac{2d+4}{d+4} \log M \gtrsim \frac{B^3(B+4dL)}{d+1}, \\
        M \gtrsim \frac{\tau}{\exp(2(d+1)M^{d/(d+4)})\delta}, \qquad M^\frac{d}{d+4} \gtrsim \frac{B^2}{d^2L^2}\log \frac{2\tau}{\delta}.
    \end{gather*}
\end{lemma}

\begin{proof}
    Combining \cref{lem:fqi-error} and \cref{lem:approx-error-affine}, we have with probability at least $1-\delta$ that the total approximation error of the output $\widetilde Q_z^\tau$ to the optimal value function $Q_z^\star$ under any measure $\widetilde\mu$ is upper bounded by
    \begin{align*}
        \|\widetilde Q_z^\tau - Q_z^\star\|_{\widetilde\mu} & \leq \sum_{\ell=1}^\tau \gamma^{\tau-\ell} \sqrt{D}\cdot \|\widetilde Q_z^\ell - \calT \widetilde Q_z^{\ell-1}\|_{\widetilde\mu} + \frac{\gamma^\tau}{1-\gamma} \\
        & \lesssim \frac{1}{1-\gamma} \sqrt{D}\cdot \Big((d+1)B^4 M^{-4/(d+4)} \log M + d^2L^2 M^{-4/(d+4)} \\
        & \qquad + BdL M^{-4/(d+4)} \sqrt{\log\frac{2\tau}{\delta}} + \frac{B^2\log\frac{2\tau}{\delta}}{M} + \epsilon\Big) + \frac{\gamma^\tau}{1-\gamma} \\
        & \lesssim \frac{1}{1-\gamma} \sqrt{D}\cdot \Big((d+1)B^4 M^{-4/(d+4)} \log M + 2d^2L^2 M^{-4/(d+4)} \\
        & \qquad + BdL M^{-4/(d+4)} \sqrt{\log\frac{2\tau}{\delta}} + \epsilon\Big) + \frac{\gamma^\tau}{1-\gamma},
    \end{align*}
    for $M \gtrsim \frac{\tau}{\exp(2(d+1)M^{d/(d+4)})\delta}$ and $M^\frac{2d+4}{d+4} \log M \gtrsim \frac{B^3(B+4dL)}{d+1}$. The first inequality follows from \cref{lem:fqi-error}, the second inequality is due to applying \cref{lem:approx-error} such that the upper bound on $\|\widetilde Q_z^\ell - \calT \widetilde Q_z^{\ell-1}\|_{\widetilde\mu}$ holds with probability at least $1-\delta/\tau$ for all $\ell\in[\tau]$, and the last inequality follows from $M^{d/(d+4)} \geq \frac{B^2}{d^2L^2}\log \frac{2\tau}{\delta}$.
    Take the number of iteration 
    \begin{align*}
        \tau & = \left\lceil\log \frac{M^{4/(d+4)}}{d^2L^2\sqrt{D}} / \log\frac{1}{\gamma}\right\rceil,
    \end{align*}
    it follows that for $M \geq \frac{1}{\delta\log\frac{1}{\gamma}} \log \frac{M^{4/(d+4)}}{d^2L^2\sqrt{D}}$, the approximation error
    \begin{align*}
        \|\widetilde Q_z^\tau - Q_z^\star\|_{\widetilde\mu} & \lesssim \frac{\sqrt{D}}{1-\gamma} \left( (d+1)B^4 + (B+dL)dL \right) M^{-\frac{4}{d+4}}\log M + \frac{\sqrt{D}\epsilon}{1-\gamma}.
    \end{align*}
\end{proof}

\subsection{Proof of FQI With ICNN}

In this subsection, we present our theoretical work for the \ICNN function family. 
As a convex object, \ICNN function is robust against outliers and input perturbations~\citep{min2021curious}, such merits have been discussed in related works~\citep{christmann2007consistency,blanchet2019multivariate,chen2020more,pfrommer2023asymmetric}. 
We combine \cref{lem:exact-representation} and \cref{lem:covering-affine} to prove \cref{lem:approx-error-nn} and then \cref{lem:fqi-approx-nn}. 
More specifically, \cref{lem:covering-affine} provides an upper bound on the covering entropy for $L$-Lipschitz $K$-piece max-affine functions, and the covering entropy of such function set can provide an upper bound on the covering entropy of $\scrN_{B,L}^K$ as $\scrN_{B,L}^K\subseteq \scrA_{B,L}^{K+1}$.
We note that our argument assumes finding the global minimum of the functions represented by the \ICNN. 
We do not elaborate on the convergence properties of such neural network functions under different backpropagation schemes, and they have been discussed in a vast body of literature \citep{chen2018optimal,schmidt2011convergence,arora2019fine,du2019graph,chen2021multiple,song2021convergence,tan2022data}. 

\begin{lemma}\label{lem:approx-error-nn}
    Let $B \defeq 1/(1-\gamma)$ and $L \defeq L_r/(1-\gamma)$, and for any mean-field term $z$, $Q\in\scrC_{B,L}$, and data $\calD_{\calT_z Q} \defeq \{(s_m,a_m,r_m,s'_m)\}_{m=1}^M$ with $(s_m,a_m)\in\Rd$ \iid sampled from any measure $\widetilde\mu\in\Delta_\calG$ for all $m\in[M]$, we define the $\epsilon$-approximate estimator $\widehat f\in\mathrm{LSE}(\scrN_{B,L}^K,\calD_{\calT_z Q};\epsilon)$ of $\calT_z Q$ for $K\in\ZZ_+$. Then for any training set size $M \defeq |\calD_{\calT_z Q}|$ such that
    \begin{align*}
        M \gtrsim \frac{1}{e^{2(d+1)K}\delta}, \qquad M^\frac{2d+4}{d+4} \log M \gtrsim \frac{B^3(B+4dL)}{d+1},
    \end{align*}
    it holds with probability at least $1-\delta$
    \begin{align*}
        \|\widehat f - \calT_z Q\|_{\widetilde\mu}^2 & \lesssim (d+1)B^4 M^{-4/(d+4)} \log M + d^2L^2 M^{-4/(d+4)} \\
        & \qquad + BdL M^{-4/(d+4)} \sqrt{\log\frac{2}{\delta}} + \frac{B^2\log\frac{2}{\delta}}{M} + \epsilon.
    \end{align*}
\end{lemma}

\begin{proof}
    Let us denote $g_z^\dagger(Q) \defeq \argmin_{f\in\scrN_{B,L}^K} \calL(f,\calT_z Q)$ to be the estimator of $\calT_z Q$ in $\scrN_{B,L}^K$ that minimizes the true risk under $\widetilde\mu$ for any mean-field term $z$ and value function $Q\in\scrC_{B,L}$. 
    We also define $g_z(Q;\calD) \defeq \argmin_{f\in\scrN_{B,L}^K} \calL_\calD(f,\calT_z Q)$ to be the \LSE of $\calT_z Q$ in $\scrN_{B,L}^K$ that minimizes the empirical risk on $\calD$. 
    For any $\widehat f\in\mathrm{LSE}(\scrN_{B,L}^K,\calD;\epsilon)$, it follows by definition that $\calE_{\calD}(\widehat f,\calT_z Q) = \calL_\calD(\widehat f,\calT_z Q) - \calL_\calD(g_z(Q;\calD),\calT_z Q) \leq \epsilon$.
    
    Recall that $\|\widehat f - \calT_z Q\|_{\widetilde\mu}^2 = \calL(\widehat f,\calT_z Q) - \calL(\calT_z Q,\calT_z Q)$.
    Following \cref{thm:empirical-bound}, it holds that for any $\alpha,\beta > 0$ and $0 < \varrho \leq \frac{1}{2}$,
    \begin{align*}
        & \prob\Big\{\exists f \in \scrN_{B,L}^K: \calL(f,\calT_z Q) - \calL(\calT_z Q,\calT_z Q) - (\calL_\calD(f,\calT_z Q) - \calL_\calD(\calT_z Q,\calT_z Q)) \\
        & \qquad\qquad\qquad \geq \varrho \cdot(\alpha+\beta+\calL(f,\calT_z Q) - \calL(\calT_z Q,\calT_z Q))\Big\}, \\
        & \qquad \leq 14 \sup_\calD \calN_1\left(\frac{\beta \varrho}{20B}, \scrN_{B,L}^K, \calD\right) \exp \left(-\frac{\varrho^2(1-\varrho) \alpha B^4 M}{214(1+\varrho)}\right),
    \end{align*}
    where the inequality follows from $|f| \leq B+dL = B$ for all $f\in\scrN_{B,L}^K$, and $\calN_1\left(\eps, \scrF, \calD\right)$ denotes the $\eps$-covering number of $\scrF$ on $\calD$ with respect to $\ell_1$ metric. 
    The supremum is taken with respect to all possible data $\calD$.
    This provides a high probability upper bound on $\|\widehat f - \calT_z Q\|_{\widetilde\mu}^2$ as 
    \begin{align*}
        & \prob\Big\{(1-\varrho)[\calL(\widehat f,\calT_z Q) - \calL(\calT_z Q,\calT_z Q)] \geq \varrho \cdot(\alpha+\beta) + [\calL_\calD(\widehat f,\calT_z Q) - \calL_\calD(\calT_z Q,\calT_z Q)]\Big\} \\
        & \qquad \leq \prob\Big\{\exists f \in \scrN_{B,L}^K: \calL(f,\calT_z Q) - \calL(g^\dagger,\calT_z Q) - (\calL_\calD(f,\calT_z Q) - \calL_\calD(g^\dagger,\calT_z Q)) \\
        & \qquad\qquad\qquad \geq \varrho \cdot(\alpha+\beta+\calL(f,\calT_z Q) - \calL(g^\dagger,\calT_z Q))\Big\}.
    \end{align*}
    Further, for the set of all functions in $\scrN_{B,L}^K$, its covering entropy is bounded by 
    \begin{align*}
        \sup_\calD\log \calN_1\left(\frac{\beta \varrho}{20B}, \scrN_{B,L}^K, \calD\right) & \leq \log \calN_\infty\left(\frac{\beta \varrho}{20B}, \scrA_{B,L}^{K+1}\right) \\
        & \leq (d+1) (K+1) \log\frac{(40B+80dL)B}{\beta \varrho},
    \end{align*}
    where the first inequality follows from upper bounding $\ell_1$ covering number with $\ell_\infty$ covering number and \cref{lem:exact-representation}, and the second inequality follows from \cref{lem:covering-affine} and $L_d \leq dL$. More specifically, for Lipschitz constant $L$ uniform $\ell_\infty$ bound $B$, the covering entropy of $\scrA_{B,L}^{K+1}$ is bounded through
    \begin{align*}
        \log\calN_\infty(\eps,\scrA_{B,L}^K) \leq (d+1) (K+1) \log\frac{B+4L_d}{\eps}
    \end{align*}
    for any $\eps \leq B+4L_d$.
    Thus, for any $\beta \leq \frac{20}{\varrho}(4dL_r + B)B$, we have
    \begin{align*}
        & \prob\Big\{(1-\varrho)[\calL(\widehat f,\calT_z Q) - \calL(\calT_z Q,\calT_z Q)] \geq \varrho \cdot(\alpha+\beta) + [\calL_\calD(\widehat f,\calT_z Q) - \calL_\calD(\calT_z Q,\calT_z Q)]\Big\} \\
        & \qquad \leq 14 \exp\left((d+1) (K+1) \log\frac{(20B+80dL)B}{\beta \varrho} - \frac{\varrho^2(1-\varrho) \alpha B^4 M}{214(1+\varrho)}\right).
    \end{align*}
    Take $\varrho = \frac{1}{2}$ and for the estimator $\widehat f\in\scrN_{B,L}^K$ that with probability at least 
    \begin{align*}
        1 - 14 \exp\left((d+1) (K+1) \log\frac{(40B+160dL)B}{\beta} - \frac{\alpha B^4 M}{2568}\right),
    \end{align*}
    it holds that the estimation error is bounded by
    \begin{align*}
        \|\widehat f - \calT_z Q\|_{\widetilde\mu}^2 & = \calL(\widehat f,\calT_z Q) - \calL(\calT_z Q,\calT_z Q) \\
        & \leq \alpha+\beta + 2 (\calL_\calD(\widehat f,\calT_z Q) - \calL_\calD(\calT_z Q,\calT_z Q)) \\
        & = \alpha+\beta + 2 (\calE_{\calD}(\widehat f,\calT_z Q) + \calL_\calD(g_z(Q;\calD),\calT_z Q) - \calL_\calD(\calT_z Q,\calT_z Q)) \\
        & \leq \alpha + \beta + 2\epsilon + 2(\calL_\calD(g_z(Q;\calD),\calT_z Q) - \calL_\calD(\calT_z Q,\calT_z Q)) \\
        & \leq \alpha + \beta + 2\epsilon + 2(\calL_\calD(g_z^\dagger(Q),\calT_z Q) - \calL_\calD(\calT_z Q,\calT_z Q)),
    \end{align*}
    where the second equality is due to $\calE_{\calD}(\widehat f,\calT_z Q) = \calL_\calD(\widehat f,\calT_z Q) - \calL_\calD(g_z(Q;\calD),\calT_z Q)$ and the last inequality is due to the fact that $g_z(Q;\calD)$ is defined to be the \LSE that minimizes the empirical risk, \ie, $\calL_\calD(g_z(Q;\calD),\calT_z Q) \leq \calL_\calD(g_z^\dagger(Q),\calT_z Q)$.
    Consequently, solve for proper $\alpha$ and $\beta$ to get that if  
    \begin{align*}
        M \gtrsim \frac{1}{e^{2(d+1)K}\delta}, \qquad M\log M \gtrsim \frac{B+4dL}{B^3(d+1)K},
    \end{align*}
    then for any $d\geq 2$, with probability at least $1-\delta/2$ 
    \begin{align*}
        \|\widehat f - \calT_z Q\|_{\widetilde\mu}^2 & \lesssim (d+1)(K+1)B^4 \frac{\log M}{M} + 2\epsilon + 2(\calL_\calD(g_z^\dagger(Q),\calT_z Q) - \calL_\calD(\calT_z Q,\calT_z Q)).
    \end{align*}
    
    Following \cref{lem:approx-affine} and the conclusion that $\scrA_{B,L}^{K,+}\subseteq\scrN_{B,L}^K$ from \cref{lem:exact-representation}, under any measure $\widetilde\mu$ it holds that
    \begin{align*}
        \|g_z^\dagger(Q) - \calT_z Q\|_{\widetilde\mu}^2 \leq \|\overline g_z^\ddagger(Q) - \calT_z Q\|_{\widetilde\mu}^2 \leq 72^2 L_d^2 (K-1)^{-4/d},
    \end{align*}
    where $\overline g_z^\ddagger(Q) \defeq \sigma(g_z^\ddagger(Q))$ is the truncated estimator, where $\sigma$ denotes ReLU function and $g_z^\ddagger(Q) \defeq \argmin_{f\in\scrA_{B,L}^{K-1}} \calL(f,\calT_z Q)$.
    Note that $\|g_z^\dagger(Q) - \calT_z Q\|_{\widetilde\mu}^2 = \calL(g_z^\dagger(Q),\calT_z Q) - \calL(\calT_z Q,\calT_z Q)$, and following Bernstein's inequality, we have
    \begin{align*}
        & |\calL(g_z^\dagger(Q),\calT_z Q) - \calL(\calT_z Q,\calT_z Q) - \calL_\calD(g_z^\dagger(Q),\calT_z Q) + \calL_\calD(\calT_z Q,\calT_z Q)| \\
        & \qquad \leq \frac{2B^2\log\frac{1}{\delta}}{3M} + 4B \|g_z^\dagger(Q) - \calT_z Q\|_{\widetilde\mu} \sqrt{\frac{\log\frac{1}{\delta}}{M}}.
    \end{align*}
    It follows that with probability at least $1-\delta/2$ 
    \begin{align*}
        \calL_\calD(g_z^\dagger(Q),\calT_z Q) - \calL_\calD(\calT_z Q,\calT_z Q) & \leq 72^2 d^2L^2 (K-1)^{-4/d} + 288BdL K^{-2/d} \sqrt{\frac{\log\frac{2}{\delta}}{M}} + \frac{2B^2\log\frac{2}{\delta}}{3M}.
    \end{align*}
    
    Combine the upper bounds above together, we have with probability at least $1-\delta$
    \begin{align*}
        \|\widehat f - \calT_z Q\|_{\widetilde\mu}^2 & \lesssim \frac{(d+1)(K+1)}{(1-\gamma)^4} \frac{\log M}{M} + d^2L^2 (K-1)^{-4/d} + BdL K^{-2/d} \sqrt{\frac{\log\frac{2}{\delta}}{M}} + \frac{B^2\log\frac{2}{\delta}}{M} + \epsilon.
    \end{align*}
    Take $K = \lceil M^{d/(d+4)}\rceil$, then it holds that
    \begin{align*}
        \|\widehat f - \calT_z Q\|_{\widetilde\mu}^2 & \lesssim (d+1)B^4 M^{-4/(d+4)} \log M + d^2L^2 M^{-4/(d+4)} \\
        & \qquad + BdL M^{-4/(d+4)} \sqrt{\log\frac{2}{\delta}} + \frac{B^2\log\frac{2}{\delta}}{M} + \epsilon.
    \end{align*}
\end{proof}

\begin{lemma}\label{lem:fqi-approx-nn}
    Let $B \defeq 1/(1-\gamma)$ and $L \defeq L_r/(1-\gamma)$.
    For the \MDP induced by any mean-field term $z$, if \cref{algo:q-iteration} achieves $\epsilon$-approximate \LSE in $\scrN_{B,L}^K$ with $K = \lceil M^{d/(d+4)}\rceil$ at all $\tau$ iterations, then the approximation error of $\widetilde Q_z^\tau$ with respect to the optimal value function $Q_z^\star$ under any measure $\widetilde\mu$ is upper bounded by
    \begin{align*}
        \|\widetilde Q_z^\tau - Q_z^\star\|_{\widetilde\mu} & \lesssim \frac{\sqrt{D}}{1-\gamma} \left( (d+1)B^4 + (B+dL)dL \right) M^{-\frac{4}{d+4}}\log M + \frac{\sqrt{D}\epsilon}{1-\gamma}
    \end{align*}
    with probability at least $1-\delta$ for the number of iterations 
    \begin{align*}
        \tau \asymp \log \frac{M^{4/(d+4)}}{d^2L^2\sqrt{D}} / \log\frac{1}{\gamma}
    \end{align*}
    for any sample size $M = |\calD|$ such that
    \begin{gather*}
        M \gtrsim \frac{1}{\delta\log\frac{1}{\gamma}} \log \frac{M^{4/(d+4)}}{d^2L^2\sqrt{D}}, \qquad M^\frac{2d+4}{d+4} \log M \gtrsim \frac{B^3(B+4dL)}{d+1}, \\
        M \gtrsim \frac{\tau}{\exp(2(d+1)M^{d/(d+4)})\delta}, \qquad M^\frac{d}{d+4} \gtrsim \frac{B^2}{d^2L^2}\log \frac{2\tau}{\delta}.
    \end{gather*} 
\end{lemma}

\begin{proof}
    Combining \cref{lem:fqi-error} and \cref{lem:approx-error-affine}, we have with probability at least $1-\delta$ that the total approximation error of the output $\widetilde Q_z^\tau$ to the optimal value function $Q_z^\star$ under any measure $\widetilde\mu$ is upper bounded by
    \begin{align*}
        \|\widetilde Q_z^\tau - Q_z^\star\|_{\widetilde\mu} & \leq \sum_{\ell=1}^\tau \gamma^{\tau-\ell} \sqrt{D}\cdot \|\widetilde Q_z^\ell - \calT \widetilde Q_z^{\ell-1}\|_{\widetilde\mu} + \frac{\gamma^\tau}{1-\gamma} \\
        & \lesssim \frac{1}{1-\gamma} \sqrt{D}\cdot \Big((d+1)B^4 M^{-4/(d+4)} \log M + d^2L^2 M^{-4/(d+4)} \\
        & \qquad + BdL M^{-4/(d+4)} \sqrt{\log\frac{2\tau}{\delta}} + \frac{B^2\log\frac{2\tau}{\delta}}{M} + \epsilon\Big) + \frac{\gamma^\tau}{1-\gamma} \\
        & \lesssim \frac{1}{1-\gamma} \sqrt{D}\cdot \Big((d+1)B^4 M^{-4/(d+4)} \log M + 2d^2L^2 M^{-4/(d+4)} \\
        & \qquad + BdL M^{-4/(d+4)} \sqrt{\log\frac{2\tau}{\delta}} + \epsilon\Big) + \frac{\gamma^\tau}{1-\gamma},
    \end{align*}
    for $M \gtrsim \frac{\tau}{\exp(2(d+1)M^{d/(d+4)})\delta}$ and $M^\frac{2d+4}{d+4} \log M \gtrsim \frac{B^3(B+4dL)}{d+1}$. The first inequality follows from \cref{lem:fqi-error}, the second inequality is due to applying \cref{lem:approx-error} such that the upper bound on $\|\widetilde Q_z^\ell - \calT \widetilde Q_z^{\ell-1}\|_{\widetilde\mu}$ holds with probability at least $1-\delta/\tau$ for all $\ell\in[\tau]$, and the last inequality follows from $M^{d/(d+4)} \geq \frac{B^2}{d^2L^2}\log \frac{2\tau}{\delta}$.
    Take the number of iteration 
    \begin{align*}
        \tau & = \left\lceil\log \frac{M^{4/(d+4)}}{d^2L^2\sqrt{D}} / \log\frac{1}{\gamma}\right\rceil,
    \end{align*}
    it follows that for $M \geq \frac{1}{\delta\log\frac{1}{\gamma}} \log \frac{M^{4/(d+4)}}{d^2L^2\sqrt{D}}$, the approximation error
    \begin{align*}
        \|\widetilde Q_z^\tau - Q_z^\star\|_{\widetilde\mu} & \lesssim \frac{\sqrt{D}}{1-\gamma} \left( (d+1)B^4 + (B+dL)dL \right) M^{-\frac{4}{d+4}}\log M + \frac{\sqrt{D}\epsilon}{1-\gamma}.
    \end{align*}
\end{proof}

\subsection{Proof of \cref{lem:exact-representation}}\label{sec:exact-representation-proof}

\cref{lem:exact-representation} shows the existence of a set $\scrN_{B,L}^K$ of \ICNNs that covers the set of all non-negative $L$-Lipschitz $K$-piece max-affine functions bounded by $B$ with a convex parameter set. Under such parameter set, we show that $\scrN_{B,L}^K$ is equivalent to the set of all truncated $L$-Lipschitz $K$-piece max-affine functions bounded by $B$, which is a subset of all positive $L$-Lipschitz $(K+1)$-piece max-affine functions bounded by $B$.

\begin{proof}
    The proof is based on the construction of a $K$-layer \ICNN similar to that in \citet[Theorem 1]{chen2018optimal}, where they showed that any $K$-piece max-affine function can be represented exactly by a $K$-layer \ICNN. 
    However, \citet{chen2018optimal} did not provide a parameter set for the neural network such that the $K$-layer \ICNN is able to represent non-negative $L$-Lipschitz functions bounded by $B$ exclusively, which is what we are going to show in this proof.
    
    Especially, $\scrN_{B,L}^K$ is a set of $K$-layer \ICNNs with input dimension $2d$ and all hidden layer dimension $1$. 
    For any input $x\in\calX$ and $L$-Lipschitz $K$-max-affine function $f_K(x) = \max\{\alpha_1\transpose x + c_1,\ldots,\alpha_K\transpose x + c_K\}$ that maps from $\calX$ to $[0,B]$, the corresponding input vector for the network representation of $f_K$ in $\scrN_{B,L}^K$ is given by $x^\ddagger \defeq [x\transpose,-x\transpose]\transpose$.
    More specifically, we have all layers $y_i$ of the network representation of $f_K$ defined as
    \begin{align*}
        y_i = \sigma(y_{i-1} + (\alpha_{i}-\alpha_{i+1})_+\transpose x - (\alpha_{i+1}-\alpha_{i})_+\transpose x + (c_{i}-c_{i+1}))
    \end{align*}
    for all $i = 2,\ldots,K-1$ and
    \begin{gather*}
        y_1 = \sigma((\alpha_{1}-\alpha_{2})_+\transpose x - (\alpha_{2}-\alpha_{1})_+\transpose x + (c_{1}-c_{2})), \\
        y_K = \sigma(y_{K-1} + (\alpha_{K})_+\transpose x - (-\alpha_{K})_+\transpose x + c_{K}).
    \end{gather*}
    Notice that we have $W_0^{(y)} = 0$ and $W_i^{(y)} = 1$ fixed for all $i\in[K-1]$; the shifting terms $\beta_{K-1} = c_K$ and $\beta_i = c_{i+1} - c_{i+2}$ for all $i = 0,\ldots,K-2$. Further, we have for all $i = 0,\ldots,K-2$
    \begin{align*}
        W_i^{(x)} = [(\alpha_{i+1}-\alpha_{i+2})_+\transpose, (\alpha_{i+2}-\alpha_{i+1})_+\transpose]\transpose \in\RR_+^{2d}
    \end{align*}
    and $W_{K-1}^{(x)} = [(\alpha_K)_+\transpose,(-\alpha_K)_+\transpose]\transpose\in\RR_+^{2d}$, such that $W_i^{(x)}{}\transpose x^\ddagger = (\alpha_{i+1}-\alpha_{i+2})_+\transpose x - (\alpha_{i+2}-\alpha_{i+1})_+\transpose x$ and $\|W_i^{(x)}\|_\infty = \|\alpha_{i+1}-\alpha_{i+2}\|_\infty$ for all $i = 0,\ldots,K-2$ as well as $W_{K-1}^{(x)}{}\transpose x^\ddagger = (\alpha_K)_+\transpose x - (-\alpha_K)_+\transpose x$ and $\|W_{K-1}^{(x)}\|_\infty = \|\alpha_K\|_\infty$.
    
    On the other hand, for any set of parameters $\{(W_i^{(x)},\beta_i)\}_{i=0}^{K-1}$ with $W_i^{(x)}\in\RR_+^{2d}$ and $\beta_i\in\RR$ for all $i = 0,\ldots,K-1$, 
    we deem $W_i^{(x)} = [w_{i,1}\transpose,w_{i,2}\transpose]\transpose$ 
    as concatenations of vectors $w_{i,1},w_{i,2}\in\RR^d$ 
    , then for any input $x^\ddagger$, we have 
    \begin{align*}
        \inner{W_i^{(x)}}{x^\ddagger} = w_{i,1}\transpose x - w_{i,2}\transpose x. 
    \end{align*}
    Further notice that for a concise representation $w_i \defeq w_{i,1} - w_{i,2}$,
    we have $w_i\transpose x = \inner{W_i^{(x)}}{x^\ddagger}.$
    There exists a set of vectors $\{\alpha_i\}_{i=1}^K$ such that 
    we have $\alpha_K = w_{K-1}$ and $\alpha_i = w_{i-1} + \alpha_{i+1}$ for all $i\in[K-1]$; similarly, there exists a set of scalars $\{c_i\}_{i=1}^K$ such that $c_K = \beta_{K-1}$ and $c_i = \beta_{i-1} + c_{i+1}$ for all $i\in[K-1]$.
    Solve the system of equations to get $\alpha_i = \sum_{j=i-1}^{K-1} w_j$ and $c_i = \sum_{j=i-1}^{K-1} \beta_j$ for all $i\in[K]$, and the \ICNN with parameterization $\{(W_i^{(x)},\beta_i)\}_{i=0}^{K-1}$ represents function $f_K \defeq \sigma(\max\{\alpha_1\transpose x + c_1,\ldots,\alpha_K\transpose x + c_K\})$.
    Under the convex constraint set $\|\sum_{j=i-1}^{K-1} w_j\|_\infty \leq L$, it holds that the represented max-affine function $f_K$ is $L$-Lipschitz. Moreover, the represented max-affine function $f_K$ is upper bounded by $B$ on $\calX$ if and only if the composing affine functions $\alpha_i\transpose x + c_i$ for all $i\in[K]$ are upper bounded; this is due to $f_K \geq \alpha_i\transpose x + c_i$ for all $i\in[K]$ by definition.
    It requires that $\alpha_i\transpose x + c_i \leq B$ for all $i\in[K]$ and $x\in\calX$, which can be equivalently written as
    \begin{align*}
        \max_{x\in\calX} \sum_{j=i-1}^{K-1} w_j\transpose x + \sum_{j=i-1}^{K-1} \beta_j \leq B.
    \end{align*}
    Note that for any $\alpha_1,\alpha_2\in\Rd$, it holds that $\max_{x\in\calX} (\alpha_1\transpose x + \alpha_2\transpose x) \leq \max_{x\in\calX} \alpha_1\transpose x + \max_{x\in\calX} \alpha_2\transpose x$,
    and the boundedness requirement forms a convex constraint set following the convexity of $\calX$. It is not hard to see that under the convex constraint set described by all the aforementioned conditions, the set $\scrN_{B,L}^K$ \ICNN represents only the truncated $K$-max-affine functions that are $L$-Lipschitz and upper bounded by $B$.
\end{proof}

\section{Supporting Lemmata}

\begin{lemma}\label{lem:dual}
    Let $f: \frakF \rightarrow \mathbb{R} \cap\{\infty\}$ be a differentiable $\zeta$-strongly convex function with respect to a norm $\|\cdot \|$, where $\frakF$ is the set of all measurable functions on $\calA\subseteq\RR^{d_w}$. Let the effective domain of $f$ be $S=\{x \in \frakF: f(x)\in\RR\}$ and $f^\star$ be the Fenchel conjugate of  $f$. Then, we have
    \begin{enumerate}
        \item 
        $f^\star$ is differentiable on $\frakF$;
        \item
        $\nabla f^\star(y)=\arg \max _{x \in S}\{\langle x, y\rangle-f(x)\}$;
        \item
        $f^\star$ is $\frac{1}{\zeta}$-smooth with respect to the dual norm  $\|\cdot \|_{*}$. That is, for all $y_1, y_2 \in \frakF$, we have
        \begin{align*}
            \|\nabla f^\star(y)-\nabla f^\star(y')\| \leq \frac{1}{\zeta}\|y-y'\|_{*}.
        \end{align*}
    \end{enumerate}
\end{lemma}

\begin{proof}
    This lemma is adopted from \citet{shalev2007online}, and we provide a brief proof here for completeness.
    Suppose we assume there exist $x,x'\in S$ such that $x\in\partial f^*(y)$ and $x'\in\partial f^*(y)$ for any $y\in S$. 
    For any $x\in\partial f^*(y)$, it holds
    \begin{align*}
        f^*(g) \geq f^*(y) + \inner{y}{g-x}
    \end{align*}
    for any $g\in S$. Rearrange the terms to get 
    Since $(f^*)^* = f$ and the fact that for any $y'$
    \begin{align*}
        f^*(y') & = \sup_{g\in S} \{\inner{g}{y'} - f(g)\} \\
        & \geq \inner{x}{y'} - f(x) \\
        & = f^*(y) + \inner{x}{y'-y},
    \end{align*}
    we conclude that $y\in\partial f(x)$ and $y\in\partial f(x')$. 
    Since $f$ is $\zeta$-strongly convex function, then for any $x,x'\in\frakF$ and $y\in\partial f(x)$, $y'\in\partial f(x')$, we have $\|x-x'\|^2 \leq \frac{1}{\zeta}\inner{x-x'}{y-y'}$. Following Holder's inequality, we conclude that
    \begin{align*}
        \|\nabla f^\star(y)-\nabla f^\star(y')\| \leq \frac{1}{\zeta}\|y-y'\|_{*}.
    \end{align*}
\end{proof}

\begin{theorem}\label{thm:empirical-bound}
    We assume a data set $\{(X_i,Y_i)\}_{i=1}^n$ generated from distribution $(X,Y)\sim\scrD$ such that there exists $B > 1$ and $|Y| \leq B$ almost surely. We further define $m(x) \defeq \expect[Y\given X = x]$. Let $\scrF$ be a set of functions $f: \Rd \rightarrow \RR$ and let $|f(x)| \leq B$. Then, it holds for each $n \geq 1$ that
    \begin{align*}
        & \prob\Big\{\exists f \in \scrF: \expect|f(X)-Y|^2 - \expect|m(X)-Y|^2 - \frac{1}{n} \sum_{i=1}^{n}\{|f(X_{i})-Y_{i}|^2-|m(X_{i})-Y_{i}|^2\} \\
        & \qquad\qquad\qquad \geq \varrho \cdot(\alpha+\beta+\expect|f(X)-Y|^2-\expect|m(X)-Y|^2)\Big\} \\
        & \qquad \leq 14 \sup_{x_{1}^{n}} \calN_1\left(\frac{\beta \varrho}{20 B}, \scrF, x_{1}^{n}\right) \exp \left(-\frac{\varrho^2(1-\varrho) \alpha n}{214(1+\varrho) B^{4}}\right),
    \end{align*}
    where $x_1^n \defeq (x_1,\ldots,x_w)$ denotes $n$ fixed points in $\Rd$, $\alpha, \beta>0$, and $0<\varrho \leq 1 / 2$.
\end{theorem}

\begin{proof}
    This is \citet[Theorem 11.4]{gyorfi2002distribution}.
\end{proof}

\begin{lemma}\label{lem:covering-convex}
    Define the class of uniformly $B$-bounded, subdifferentiable, and uniformly $L$-Lipschitz functions on $\calX\subseteq\Rd$ as
    \begin{align*}
        \scrC_{\calX,B,L} & \defeq\left\{f: \calX \rightarrow \RR \mid f \text { is convex, }\|f\|_{\infty} \leq B, \forall x \in \calX: \partial f(x) \neq \emptyset, \forall g \in \partial f(x):\|g\|_{\infty} \leq L\right\}.
    \end{align*}
    Let $R_d^* \defeq \max\{8L_d, 2B+4L_d\}$ where $L_d \defeq dL\diam(\calX)$, then for all $\eps\in(0,80L_d]$, the covering entropy of $\scrC_{\calX,B,L}$ is bounded by
    \begin{align*}
        \log\calN_\infty(\eps, \scrC_{\calX,B,L}) \leq 2(d+1)\left(\frac{80 L_d}{\eps}\right)^{d / 2} \ln \left(\frac{10 R_d^*}{\eps}\right).
    \end{align*}
\end{lemma}

\begin{proof}
    This is \citet[Lemma 4.3]{balazs2015near}.
\end{proof}

\begin{lemma}\label{lem:covering-affine}
    For any $K\in\ZZ_+$, define the class of uniformly $B$-bounded and $L$-Lipschitz $K$-piece max-affine functions on $\calX\subseteq\Rd$ as
    \begin{align*}
        \scrA_{\calX,B,L}^K \defeq \{h:\calX\to\RR \given h(x) = \max_{k\in[K]} \alpha_k\transpose x + c_k, \|\alpha_k\|_\infty \leq L, h(x)\in [-B_d,B]\}
    \end{align*}
    where $B_d\defeq B+dL\diam(\calX)$.
    Then for any $R_d \defeq 2B+4L_d$ with $L_d \defeq dL\diam(\calX)$ and any $\eps\in(0, R_d]$,
    \begin{align*}
        \log\calN_{\infty}(\eps, \scrA_{\calX,B,L}^K) \leq (d+1) K \log\frac{R_d}{\eps}.
    \end{align*}
\end{lemma}

\begin{proof}
    This is \citet[Lemma 4.2]{balazs2015near}.
\end{proof}

\begin{lemma}\label{lem:approx-affine}
    For all $B$-bounded and $L$-Lipschitz convex function $f\in\scrC_{\calX,B,L}$ on $\calX\subseteq\Rd$ and any $K\in\ZZ_+$, it holds that the best approximation $h$ of $f$ in the set of $\scrA_{\calX,B,L}^K$ satisfies
    \begin{align*}
        \inf_{h\in\scrA_{\calX,B,L}^K}\|f-h\|_{\infty} \leq 72 L_d K^{-2/d},
    \end{align*}
    where $L_d \defeq dL\diam(\calX)$.
\end{lemma}

\begin{proof}
    This is \citet[Lemma 4.1]{balazs2015near}.
\end{proof}

\end{document}